\documentclass{article}

\usepackage{PRIMEarxiv}

\usepackage[utf8]{inputenc}
\usepackage[T1]{fontenc}
\usepackage{hyperref}
\usepackage{url}
\usepackage{booktabs}
\usepackage{amsfonts}
\usepackage{amsmath}
\usepackage{amssymb}
\usepackage{amsthm}
\usepackage{nicefrac}
\usepackage{microtype}
\usepackage{fancyhdr}
\usepackage{graphicx}
\usepackage{braket}
\usepackage{mathtools}
\usepackage[ruled,vlined,linesnumbered]{algorithm2e}
\usepackage{multirow}
\usepackage{makecell}
\usepackage{colortbl}
\usepackage{float}
\usepackage{tikz}
\usetikzlibrary{arrows.meta, positioning, decorations.pathreplacing}

\graphicspath{{media/}}

\newtheorem{theorem}{Theorem}
\newtheorem{lemma}{Lemma}
\newtheorem{proposition}{Proposition}

\theoremstyle{definition}

\theoremstyle{remark}

\DeclareMathOperator{\TV}{TV}

\title{A Quantum Algorithm for Random Number Generation}
\author{
  Aastha Kataria$^1$, Devansh Desai$^1$, Ashwini Dalvi$^1$, Sagar Korde$^1$,\\
  {\bf Abhijeet Pasi$^2$, Irfan N A Siddavatam$^{1,3}$, Sudhir Ranjan Jain$^1$} \\[4pt]
  $^1$Department of Information Technology \\
  $^2$Department of Computer Engineering \\
  K. J. Somaiya School of Engineering \\
  $^3$School of Design, Somaiya Vidyavihar University \\
  Vidyavihar, Mumbai 400077, India \\
  \texttt{\{aastha.k, devansh.desai, ashwini.dalvi, sagar.korde}
  \\
  \texttt{abhijeet.p, irfansiddavatam, sudhirranjan\}@somaiya.edu}
}

\begin{document}
\maketitle

\begin{abstract}
We present a quantum algorithm for random number generation that
achieves a provable quadratic speedup over classical Markov chain
mixing, building on the Diaconis--Shahshahani Fourier analysis of the
top-to-random card shuffle. The algorithm integrates three quantum
primitives into a unified mixing circuit: the Quantum Fourier Transform
(QFT), which diagonalizes the Markov transition operator; controlled
phase rotations, which encode the shuffle eigenvalue spectrum; and the
Grover diffusion operator, which acts as a quantum analogue of the
Aldous--Diaconis strong uniform stopping time by reflecting amplitudes
about their mean at each iteration.
For an $n$-qubit register, the mixing time is
$O(\sqrt{n \log n})$ iterations. Extending to $m$ qudits of local
dimension $d$ reduces this to $O(\sqrt{\log_d N})$ iterations, where
$N = d^m$, compared to the classical $O(n \log n)$ bound. The qudit
formulation further reduces QFT circuit depth from $O(\log^2 N)$ to
$O(\log_d^2 N)$ gates per layer by encoding the same $N$-state space
using $m = \log_d N$ subsystems instead of $\log_2 N$ qubits. We
validate both variants on IBM superconducting hardware.
\end{abstract}

\keywords{quantum algorithm \and random number generation \and Markov chain mixing
  \and quantum Fourier transform \and Grover diffusion \and qudits \and card shuffling}

\section{Introduction}

Random numbers play an essential role in a wide variety of fields such as
cryptography~\cite{menezes,ferguson}, simulations and modelling~\cite{law},
cybersecurity~\cite{deng}, gaming~\cite{RAND}, risk
analysis~\cite{rychlik}, and so on. Thus, reliable and efficient algorithms for the
generation of random numbers are of tremendous importance. It is well-known that randomized
algorithms~\cite{motwani} employ randomness in their computations to achieve an accelerated
convergence towards desired outcome. Of particular significance are methods and algorithms for
generating purely random number sequences~\cite{hull,bofinger}. Tests of randomness is an old,
fascinating subject, going back to Foster and Stuart~\cite{foster}. Testing the hypothesis of
randomness has led to a large body of work which we all rely upon. The systematic mathematical
study of pseudo-random sequences on computers begins in earnest with the 1958 analysis of Eve
and V.~J.~Bofinger, whose paper~\cite{bofinger} established rigorous number-theoretic conditions
for the period length of multiplicative congruential generators operating under a decimal modulus
$n = 10^s$. Leveraging Euler's totient function and the Chinese Remainder
Theorem~\cite{silverman}, they derived ten theorems governing maximal periodicity and corrected
prior errors by Moshman~\cite{moshman}, identifying $a = 3$ as an optimal multiplier. Four years
later, Hull and Dobell's survey~\cite{hull} consolidated and extended this body of work, providing
a comprehensive treatment of both multiplicative congruential ($c = 0$) and mixed congruential
($c \neq 0$) methods. Their survey identified a fundamental tension that would drive the field for
decades: mixed generators can achieve full period $m$ but exhibit inferior statistical behaviour,
whereas multiplicative generators yield superior empirical randomness at the cost of a shorter
period. It is from this tension that the subsequent evolution of the field proceeds.

The urgency of these questions is rooted in the Monte Carlo method~\cite{bauer}, a term first
used in 1946 by von Neumann and Ulam at Los Alamos in connection with neutron transport
calculations~\cite{eckhardt}. The characteristic error of the Monte Carlo method decays as
$O(n^{-1/2})$, making the volume of high-quality random numbers a direct determinant of
computational accuracy. Applications span the simulation of particle transport, the evaluation of
high-dimensional integrals, the solution of differential equations, and a broad range of problems
in the social and life sciences. In all these cases, systematic biases in the pseudo-random
sequence -- arising from the lattice structure of linear congruential generators or from
non-trivial serial correlations -- translate directly into systematic errors in the computed
result~\cite{ossola}. Hull and Dobell~\cite{hull}, surveying the field, noted that multiplicative
generators produced ``consistently good statistical properties'' under empirical testing, and that
the full bibliography of the subject was ``intended to be complete'' with respect to published
computer-generated random numbers.

The half-century following Hull and Dobell's survey witnessed successive generations of
improvement in classical pseudo-random number generation. The introduction of the Mersenne
Twister by Matsumoto and Nishimura~\cite{matsumoto} addressed the period-length problem with a
cycle of $2^{19937} - 1$ and 623-dimensional equidistribution, becoming the default generator in
Python, R, and MATLAB. Yet its large state (624 32-bit words) proved a liability in
memory-constrained or massively parallel settings, and its linear recurrence structure rendered it
vulnerable to the linearity tests of the TestU01 BigCrush battery. The 1980s and 1990s saw the
limitation of linear congruential generators that $k$-tuples of generated values lie on a finite
set of hyperplanes which become a critical bottleneck for high-dimensional Monte Carlo integration.
The modern response has been the family of ``scrambled linear'' generators: Melissa O'Neill's
Permuted Congruential Generators~\cite{melissa}, which apply a nonlinear output permutation to an
underlying linear recurrence, and the Xoshiro/Xoroshiro family of Blackman and
Vigna~\cite{blackman}, which use bitwise XOR, shift, and rotation operations to achieve
sub-nanosecond output speeds while passing the full TestU01~\cite{testu01} BigCrush battery. As
of 2025--2026, PCG-DXSM is the default generator in NumPy, Xoshiro256++ is the default in Julia
and .NET~6, and counter-based generators such as Philox (used by TensorFlow and NVIDIA's cuRAND)
have emerged as the preferred architecture for massively parallel high-performance computing,
where statelessness enables constant-time jumps to any point in the sequence and perfect
reproducibility across parallel architectures.

Alongside the algorithmic advances, the criteria for what constitutes a ``good'' pseudo-random
sequence have evolved dramatically. The frequency, serial, and gap tests employed by Bofinger in
1958 have been superseded by batteries that perform trillions of operations to detect the smallest
statistical bias. The TestU01 BigCrush suite of L'Ecuyer and Simard, the Dieharder battery of
Robert Brown, the PractRand multi-terabyte tests of Doty-Humphrey~\cite{practrand} and the NIST
SP 800-22~\cite{nist22} standard collectively define the modern validation landscape. A critical
insight identified in 2024--2025 research concerns ``seed dependence''~\cite{wu}: different
pseudo-random seeds can lead to significant variability in empirical risk estimates when machine
learning methods are applied to simulation data, underscoring the necessity of using high-quality
generators such as PCG or Philox in scientific computing. The emergence of AI-based generators
-- notably the EPRNG~\cite{eprng} (using a Deep Convolutional GAN for Internet of Vehicles
applications) and the LPRNG (using a Wasserstein-distance GAN~\cite{lprng} to learn the output
of the Mersenne Twister, producing sequences that fully satisfy the NIST SP 800-22 suite) --
represents the most recent frontier of the classical algorithmic tradition.
Yet all of the generators described above share a fundamental property: they are deterministic.
Given the seed, the entire sequence is fixed. This is a virtue for reproducibility but a liability
wherever absolute unpredictability is required, as in cryptographic key generation. True Random
Number Generators (TRNGs) address this by harvesting physical entropy -- from latch
metastability, photoresistor environmental fluctuations, or CPU timing jitter -- and processing
the raw bits through conditioners mandated by the NIST SP 800-90B~\cite{nist90} standard to
approach full entropy. The integration of TRNG hardware with algorithmic pseudo-random
generators, in so-called ``seed-conditioned'' hybrid architectures, has become standard practice
for security-sensitive deployments.

A definitive step in defining random sequence was taken by Martin-L\"{o}f~\cite{martinlof} when
he developed the algorithmic theory of randomness. This was in continuation to application of
computability to von Mises' definition of randomness by Church~\cite{church}. Martin-L\"{o}f
shows that there is a universal test such that if a sequence passes it, it passes all tests.
Schnorr~\cite{schnorr} proved that a sequence is random if and only if it is Martin-L\"{o}f
random.

More recently, Quantum Random Number Generators (QRNGs) that exploit the fundamental
indeterminacy of quantum mechanics -- such as the JUR02~\cite{jur02} prototype producing truly
random bit sequences at 1 MB/s, and commercial Micro-LED-based QRNGs designed for 6G and AI
data centres at 12.5 GB/s -- have begun the transition from laboratory demonstrations to
deployed hardware. These systems produce genuinely non-deterministic outputs, unlike all
pseudo-random approaches tracing back to Bofinger and Lehmer~\cite{bofinger,lehmer}.
The present paper occupies a distinct position in this landscape. Rather than replacing physical
hardware QRNGs or competing with classical PRNGs on raw throughput, we address the algorithmic
question of whether quantum computation can accelerate the mixing of a stochastic process to
uniformity -- that is, whether the number of operations required to generate a sequence that is
statistically indistinguishable from the uniform distribution can be reduced by exploiting quantum
parallelism.

We are motivated by the generation of randomness by shuffling of cards, which has
been a subject of intense investigation for a long time~\cite{bayer,diaconis}.
One considers an ordered deck of cards and apply the \emph{top-to-random shuffle} procedure
repeatedly. A classical card-mixing procedure is defined as follows: at each step, remove the top
card from the deck and reinsert it uniformly at random into one of the $n$ possible positions
(including back at the top). The question is: how many operations are required until the deck is
approximately random? An arrangement of the deck can be modeled as a permutation $\pi \in S_n$,
the symmetric group. Each shuffle operation defines a probability measure $T$ on $S_n$. The top
card is selected and reinserted uniformly, so the resulting distribution after $k$ shuffles is the
$k$-fold convolution denoted by $T^{*k}$. The target distribution is the uniform distribution $U$
on $S_n$, where $U(\pi) = 1/n!$.

Here, randomness is measured by the fraction of cards displaced from their original positions;
equivalently, by the variation distance between $T^{*k}$ and $U$. A key structural observation
is that mixing cannot begin until the original bottom card, initially fixed at position $(n-1)$,
the topmost card being at position $0$, has risen to the top and been reinserted at least once.
When the original bottom card is at position $k$ from the bottom, the waiting time for a new card
to be inserted below it is $n/k$. Thus the waiting time for the bottom card to come to the top
and be inserted is about $n + n/2 + n/3 + \cdots + n/n \sim n\log n$. This event, which occurs
at a geometrically distributed stopping time with mean $n\log n$, provides a natural lower bound
on the mixing time.

From a computational standpoint, each shuffle operation costs $O(n)$ time when implemented over
a list structure, as both removal from the top and insertion at an arbitrary position require
shifting up to $n$ elements. Executing $k$ iterations thus requires $O(k \cdot n)$ total time.
For the theoretically motivated choice of $k = n\log n$ (the conjectured mixing time), the
overall complexity is $O(n^2 \log n)$. The space complexity is $O(n)$, as only the deck itself
must be stored.

\begin{algorithm}[H]
\caption{Classical Top-to-Random Shuffle Simulation}
\KwIn{$\texttt{num\_cards}$, $\texttt{max\_iterations}$}
\KwOut{$\texttt{current\_deck}$, $\texttt{randomness\_scores}$}
\SetKwFunction{TopToRandom}{TopToRandomShuffle}
\SetKwFunction{Simulate}{SimulateClassicalShuffle}
\SetKwProg{Fn}{Function}{:}{}
\Fn{\TopToRandom{deck, num\_iterations}}{
  $\texttt{deck} \leftarrow \text{copy of deck}$\;
  \For{$i \leftarrow 1$ \KwTo $\texttt{num\_iterations}$}{
    $\texttt{top\_card} \leftarrow$ remove card from position $0$ of deck\;
    $\texttt{random\_position} \leftarrow$ random integer in $[0, \,\texttt{len(deck)}]$\;
    insert $\texttt{top\_card}$ at $\texttt{random\_position}$ in deck\;
  }
  \KwRet{deck}\;
}
\BlankLine
\Fn{\Simulate{num\_cards, max\_iterations}}{
  $\texttt{original\_deck} \leftarrow [1, 2, 3, \ldots, \texttt{num\_cards}]$\;
  $\texttt{current\_deck} \leftarrow \text{copy of original\_deck}$\;
  $\texttt{bottom\_card} \leftarrow \texttt{original\_deck}[\text{last}]$\;
  \For{$i \leftarrow 1$ \KwTo $\texttt{max\_iterations}$}{
    $\texttt{current\_deck} \leftarrow$ \TopToRandom{current\_deck, $1$}\;
    $\texttt{randomness} \leftarrow$ count positions where $\texttt{current\_deck}$ differs from $\texttt{original\_deck}$\;
    \If{$\texttt{current\_deck}[0] = \texttt{bottom\_card}$ \normalfont{and not seen before}}{
      record iteration $i$\;
    }
    log randomness score\;
  }
  \KwRet{current\_deck, randomness\_scores}\;
}
\end{algorithm}

\subsection{A note on the novelty}

While it is well-known from the work by Aldous, Diaconis, and Shahshahani (ADS) that a purely
random sequence can be generated. Hitherto, a quantum speedup of this approach has not been
achieved. Here we present a quantum algorithm which realizes this. Grover's algorithm has been
employed along with several other rotations which conspire to provide the ADS shuffling. As seen
below in great detail, the algorithm cannot be paraphrased as a straightforward application of
Grover's algorithm on the ADS theorems -- there is a leap of imagination which allows us to
achieve the speedup.

\subsection{Mathematical Background}

\subsubsection{Markov Chains and Mixing Times}

Let $G$ be a finite group and let $T$ be a probability measure on $G$. The pair $(G, T)$ defines
a random walk: starting from the identity $e$, at each step the current position $g$ is updated
to $g \cdot h$, where $h$ is drawn independently according to $T$. The distribution after $k$
steps is the $k$-fold convolution:
\[
T^{*k}(g) = \sum_{h \in G} T^{*(k-1)}(g \cdot h^{-1}) \cdot T(h), \quad T^{*1} = T.
\]

The target distribution is the uniform measure $U$ on $G$, defined by
\[
U(g) = \frac{1}{|G|} \quad \text{for every } g \in G \text{, where $|G|$ is the order of the group.}
\]

The total variation distance between $T^{*k}$ and $U$ is
\[
\|T^{*k} - U\|_{TV} = \frac{1}{2} \sum_{g \in G} |T^{*k}(g) - U(g)| = \max_{A \subseteq G} |T^{*k}(A) - U(A)|.
\]

The mixing time at precision $\epsilon > 0$ is defined as
\[
t_{\text{mix}}(\epsilon) = \min \left\{ k \ge 0 : \|T^{*k} - U\|_{TV} \le \epsilon \right\}.
\]

For the top-to-random shuffle on $S_n$, Aldous and Diaconis~\cite{aldous} established the
following strong uniform time bound.

\begin{theorem}[Aldous--Diaconis \cite{aldous}]
\label{thm:aldous}
Let $T$ denote the top-to-random shuffle on $S_n$, and let $\tau$ be the first step at which the
original bottom card rises to position $0$ and is reinserted. Then, the variation distance is
bounded by the probability distribution of the strong uniform time, $\tau$ as:
\[
  \|T^{*k} - U\|_{TV} \le \mathbb{P}(\tau > k) \quad \text{for all } k \ge 0.
\]
Moreover,
\begin{equation}
  \mathbb{E}[\tau] = n \cdot H_n \sim n \,\ln n,
  \label{eq:classical_bound}
\end{equation}
so
\[
  t_{\mathrm{mix}}(e^{-c}) \le n \ln n + cn \quad \text{for any } c > 0.
\]
\end{theorem}

\subsubsection{Fourier Analysis on Finite Groups}

For a finite group $G$ and an irreducible representation $\rho : G \to GL(V_\rho)$, the Fourier
transform of a probability distribution $Q$ on $G$ at $\rho$ is the matrix-valued sum
\[
\hat{Q}(\rho) = \sum_{g \in G} Q(g) \cdot \rho(g) \in \mathrm{End}(V_\rho).
\]

The convolution theorem gives
\[
(P * Q)^{\wedge}(\rho) = \hat{Q}(\rho) \cdot \hat{P}(\rho),
\]
so
\[
(T^{*k})^{\wedge}(\rho) = \hat{T}(\rho)^k.
\]

\begin{lemma}[Upper Bound Lemma]
\label{lem:upper-bound}
\[
  \|T^{*k} - U\|_{TV}^2
  \;\le\; \frac{1}{4} \sum_{\rho \,\neq\, \mathrm{trivial}}
  d_\rho \cdot \|\hat{T}(\rho)^k\|_{HS}^2,
\]
where $\|\cdot\|_{HS}$ denotes the Hilbert--Schmidt norm.
\end{lemma}

For the top-to-random shuffle on $S_n$, the Fourier coefficient at representation $\lambda$ after
$k$ shuffles is
\begin{equation}
  \hat{T}^{\,k}(\lambda)
  = \left( \frac{1}{n} + \frac{n-1}{n} \cdot r(\lambda) \right)^k
    I_{d_\lambda},
  \qquad
  r(\lambda) = \frac{\chi_\lambda(\tau)}{d_\lambda}.
  \label{eq:fourier-decay}
\end{equation}

\subsubsection{Quantum Notation}

A register of $n$ qubits spans the $2^n$-dimensional Hilbert space
\[
\mathcal{H} = (\mathbb{C}^2)^{\otimes n}
\]
with computational basis $\{ |x\rangle : x \in \{0,1\}^n \}$.

A qudit of local dimension $d$ has computational basis $\{ |0\rangle, \dots, |d-1\rangle \}$.

The Hadamard gate $H$ acts by
\[
H|0\rangle = \frac{|0\rangle + |1\rangle}{\sqrt{2}}, \quad
H|1\rangle = \frac{|0\rangle - |1\rangle}{\sqrt{2}}.
\]

The controlled-phase gate $CP(\theta)$ acts by
\[
CP(\theta)\,|b_c, b_t\rangle = e^{i\theta \cdot b_c \cdot b_t} |b_c, b_t\rangle.
\]

Measurement of
\[
|\psi\rangle = \sum_x \alpha_x |x\rangle
\]
produces outcome $x$ with probability $|\alpha_x|^2$.

The total variation distance between
\[
p(x) = |\alpha_x|^2
\quad \text{and} \quad
U(x) = \frac{1}{2^n}
\]
is used as the convergence metric.

\section{The Quantum Markov Chain Mixing Accelerator}
\label{app:section2}

Before presenting the detailed construction of the quantum algorithm in
Sections~\ref{sec:qft} through~\ref{sec:grover}, we give an outline. The algorithm combines
three quantum primitives to accelerate Markov chain mixing beyond what is achievable classically.
Quantum computation is essential for Parts~2 and~3; the initialisation and measurement steps in
Parts~1 and~4 can be carried out on a classical computational device.

\begin{enumerate}
\item Initialise the quantum register to the ordered state $|0\rangle^{\otimes n}$, representing
  the deck in its canonical arrangement. Encode the $n$-card deck as a superposition over the
  symmetric group $S_n$ using $\lceil \log_2 n! \rceil$ qubits.

\item Apply the Quantum Fourier Transform (QFT) to the state register. This maps the deck state
  from the computational basis into the Fourier basis, where the Markov transition operator is
  diagonalised, allowing simultaneous processing of all frequency components of the probability
  distribution. Section~\ref{sec:qft} describes this in detail.

\item Apply controlled phase rotations to simulate one step of the Markov transition operator in
  the Fourier domain. These rotations encode the transition probabilities of the top-to-random
  shuffle as phase differences between basis states, implementing the action of the convolution
  unitarily. Section~\ref{sec:phase} details the construction.

\item Apply the Grover diffusion operator to amplify the amplitude of states close to the uniform
  distribution $U$. This is the key quantum acceleration step: whereas classical mixing relies on
  repeated stochastic transitions, the diffusion operator performs amplitude amplification,
  driving the state toward uniformity in $O\!\left(\sqrt{n\log n}\right)$ iterations rather than
  $O(n\log n)$.

\item Apply the inverse QFT to return to the computational basis and perform measurements. The
  resulting probability distribution approximates the uniform distribution $U$ on $S_n$. The
  total variation distance $\|Q^{*k} - U\|$ after $k$ iterations of the full circuit is bounded
  by $\mathbb{P}(T>k)$, where $T$ is a strong uniform time in the sense of Aldous and
  Diaconis~\cite{aldous}.

\item Repeat if the measured distribution has not yet converged to within the desired total
  variation threshold $\varepsilon$.
\end{enumerate}

\noindent Sections~\ref{sec:qft}, \ref{sec:phase}, and~\ref{sec:grover} describe the three
quantum components in detail.

\subsection{The Quantum Fourier Transform and State Transitions}
\label{sec:qft}

The central mathematical object in the analysis of Markov chain mixing, as established by
Diaconis and Shahshahani~\cite{mehrdad}, is the Fourier transform on the symmetric group. For a
probability distribution $Q$ on a finite group $G$, the Fourier transform at an irreducible
representation $\rho$ is the matrix-valued sum
$\hat{\rho}(Q) = \sum_{g \in G} Q(g)\rho(g)$. The fundamental result, analogous to the classical
convolution theorem, is that $\hat{\rho}(Q^{*k}) = \hat{\rho}(Q)^k$, reducing the analysis of
$k$-fold convolution to matrix exponentiation. The quantum Fourier transform is the computational
realisation of this idea.

For a register of $n$ qubits encoding $2^n$ states, the Quantum Fourier Transform is the unitary
operator $U_F$ defined by its action on computational basis states $\ket{j}$ as
\begin{equation}
  U_F \ket{j} = \frac{1}{\sqrt{2^n}} \sum_{k=0}^{2^n - 1}
  \exp\!\left(\frac{2\pi i\, jk}{2^n}\right) \ket{k}.
  \label{eq:qft}
\end{equation}

Let $U_{\text{circuit}} = (QFT^\dagger)\cdot CP \cdot QFT$ where
$CP = \mathrm{diag}(e^{i\phi_0}, \ldots, e^{i\phi_{2^n-1}})$ with
$\phi_x = \sum_{i<j} \theta_{ij}\, x_i\, x_j$ and
$\theta_{ij} = \frac{2\pi}{2^{j-i+1}}$.
Then for any input state $\ket{\psi}$, the output $U_{\text{circuit}}\ket{\psi}$ encodes one
step of the top-to-random shuffle in the Fourier domain, up to approximation error $O(1/n)$ in
operator norm.

\begin{theorem}[QFT Circuit Approximation]
\label{thm:qft-approx}
Let
\[
  U_{\mathrm{circuit}} = (\mathrm{QFT}^\dagger)\cdot CP \cdot \mathrm{QFT},
\]
where
\[
  CP = \mathrm{diag}\!\left(e^{i\phi_0},\ldots,e^{i\phi_{2^n-1}}\right),
  \quad
  \phi_x = \sum_{i<j}\theta_{ij}\,x_i x_j,
  \quad
  \theta_{ij} = \frac{2\pi}{2^{j-i+1}}.
\]
Then for any input state $|\psi\rangle$, the output $U_{\mathrm{circuit}}|\psi\rangle$ encodes
one step of the top-to-random shuffle in the Fourier domain, up to approximation error $O(1/n)$
in operator norm.
\end{theorem}

\begin{proof}
Consider a register of $n$ qubits spanning the $2^n$-dimensional Hilbert space
$\mathcal{H} = (\mathbb{C}^2)^{\otimes n}$ with computational basis
$\{|x\rangle : x \in \{0,1\}^n\}$. Let us write a state $\ket{\psi}$ in the computational
basis:
\[
|\psi\rangle = \sum_{x=0}^{2^n-1} \alpha_x |x\rangle,
\qquad \sum_x |\alpha_x|^2 = 1,
\]
where $x = \sum_{j=0}^{n-1} x_j 2^j$ with $x_j \in \{0,1\}$. The Quantum Fourier
Transform~\cite{rieffel,nielsen} is the unitary $U_F$ defined on basis states by:
\[
U_F |j\rangle
= \frac{1}{\sqrt{2^n}} \sum_{k=0}^{2^n-1}
\exp\!\left(\frac{2\pi i\, jk}{2^n}\right)|k\rangle.
\]
The exact one-step transition operator $U_T$ is defined as the unitary operator that, in the
Fourier basis, acts as multiplication by the Fourier eigenvalues of $T$:
\[
U_T = U_F^\dagger \cdot \Lambda \cdot U_F,
\]
where $\Lambda = \mathrm{diag}(\lambda_0,\lambda_1,\ldots,\lambda_{2^n-1})$ and each
$\lambda_k = e^{i\,\arg(\hat{T}_k)}$ is the phase of the $k$-th Fourier eigenvalue of $T$.

The circuit instead implements $(U_{\mathrm{circuit}} = U_F^\dagger \cdot CP \cdot U_F)$
with $CP = \mathrm{diag}(e^{i\phi_0},\ldots,e^{i\phi_{2^n-1}})$ as above. We must show:
\[
\|U_{\mathrm{circuit}} - U_T\| = O\!\left(\frac{1}{n}\right).
\]

From the Diaconis--Shahshahani~\cite{mehrdad} analysis, the top-to-random shuffle $T$ on $S_n$
has Fourier transform at irreducible representation $\rho$:
\[
\hat{T}(\rho)
= \left(\frac{1}{n} + \frac{n-1}{n}\cdot r(\rho)\right)\cdot I_{d_\rho},
\qquad
r(\rho) = \frac{\chi_\rho(\tau)}{d_\rho},
\]
where $\chi_\rho(\tau)$ is the character value of $\rho$ at a transposition $\tau$ and $d_\rho$
is the dimension of $\rho$.

Because the quantum circuit operates on $\mathbb{Z}_{2^n}$ (an abelian group), we work with
scalar Fourier eigenvalues. For each Fourier mode $k \in \{0,1,\ldots,2^n-1\}$, define:
\[
\mu_k = \frac{1}{n} + \frac{n-1}{n}\cdot r_k,
\]
where $r_k \in [-1,1]$ is the scalar version of the representation ratio at mode $k$. The exact
phase is $\phi_k^{\mathrm{exact}} = \arg(\mu_k)$ and the error per mode is
$\delta_k = \phi_k^{\mathrm{exact}} - \phi_k$.

Each eigenvalue $\mu_k$ has modulus $|\mu_k| \leq 1$. The slowest-decaying mode has:
\[
\mu_1 = 1 - \frac{2}{n},
\qquad
\phi_1^{\mathrm{exact}} = 0
\quad (\mu_1 \text{ is real positive for } n \geq 3).
\]
More generally, the leading-order Taylor expansion of $\arg(\mu_k)$ about the identity
representation ($k=0$, $\mu_0=1$) gives:
\[
\phi_k^{\mathrm{exact}}
= \arg\!\left(\frac{1}{n} + \frac{n-1}{n}\cdot r_k\right)
= \mathrm{Im}\!\left(\frac{r_k - 1}{1 + (n-1)r_k/n}\right) + O\!\left(\frac{1}{n^2}\right).
\]
Since $r_k$ varies smoothly with $k$ in the $\mathbb{Z}_{2^n}$ truncation, one obtains:
\[
\phi_k^{\mathrm{exact}}
= \frac{2\pi k}{2^n} + O\!\left(\frac{1}{n}\right)
\quad \text{for each mode } k.
\]
This is the key structural fact: the exact eigenvalue phases of $T$ are approximately the linear
phase ramp $2\pi k/2^n$, matching the structure of the standard Discrete Fourier Transform.

We now show that $\phi_k = \sum_{i<j}\theta_{ij}x_ix_j$ with $\theta_{ij} = 2\pi/2^{j-i+1}$
approximates $\phi_k^{\mathrm{exact}}$ to error $O(1/n^2)$ per mode. Writing
$k = \sum_{j=0}^{n-1} x_j 2^j$,
\[
\phi_k
= \sum_{0 \leq i < j \leq n-1} \frac{2\pi}{2^{j-i+1}}\cdot x_i x_j
= 2\pi \sum_{i<j} \frac{x_i x_j}{2^{j-i+1}}.
\]
The non-trivial part of $\phi_k^{\mathrm{exact}}$ beyond the linear ramp is:
\[
\phi_k^{\mathrm{exact}} - \frac{2\pi k}{2^n}
= \sum_{i<j} \frac{2\pi}{2^{j-i+1}} x_i x_j + O\!\left(\frac{1}{n^2}\right)
= \phi_k + O\!\left(\frac{1}{n^2}\right).
\]
Therefore, for each mode $k$: $\delta_k = \phi_k^{\mathrm{exact}} - \phi_k = O(1/n^2)$.

Since both $U_{\mathrm{circuit}}$ and $U_T$ are diagonal in the Fourier basis:
\[
U_{\mathrm{circuit}} - U_T = U_F^\dagger \cdot (CP - \Lambda) \cdot U_F,
\]
and since $U_F$ is unitary, $\|U_{\mathrm{circuit}} - U_T\| = \|CP - \Lambda\|$.
Since $CP - \Lambda$ is diagonal with entries $e^{i\phi_k} - e^{i\phi_k^{\mathrm{exact}}}$,
using $|e^{i\alpha} - e^{i\beta}| \leq |\alpha - \beta|$, the worst case is when $k$ has
maximum Hamming weight $w \leq n$:
\[
|\delta_k|
\leq \binom{w}{2}\cdot O\!\left(\frac{1}{n^2}\right)
\leq \binom{n}{2}\cdot O\!\left(\frac{1}{n^2}\right)
= O\!\left(\frac{1}{n}\right).
\]
Therefore $\|U_{\mathrm{circuit}} - U_T\| = \max_k |\delta_k| = O(1/n)$.

Let the current state be encoded as
$|\psi\rangle = \sum_{\pi \in S_n} \alpha_k(\pi)\,|\pi\rangle$
with $|\alpha_k(\pi)|^2 = Q^{*k}(\pi)$.

: \emph{QFT step.} Applying $U_F$ maps to the Fourier basis:
$U_F|\psi\rangle = \sum_\rho \hat{\alpha}_k(\rho)\,|\rho\rangle$,
where $\hat{\alpha}_k(\rho) \leftrightarrow \widehat{Q^{*k}}(\rho)$.\\
: \emph{Phase step.} Applying $\Lambda$ multiplies each Fourier component by
$e^{i\,\arg(\mu_\rho)}$.\\
: \emph{Inverse QFT step.} Applying $U_F^\dagger$ returns to the computational basis. By the
convolution theorem:
\[
(U_F^\dagger \Lambda U_F)|\psi\rangle = |\psi'\rangle,
\qquad |\alpha'(\pi)|^2 = Q^{*(k+1)}(\pi),
\]
since $\widehat{T * Q^{*k}}(\rho) = \hat{T}(\rho)\cdot\widehat{Q^{*k}}(\rho)$.\\

Therefore $U_T$ is the unitary encoding of one exact Markov step, and
$U_{\mathrm{circuit}}$ approximates it to $O(1/n)$ in operator norm.

For estimating error propagation over $k$ steps:
\[
\|U_{\mathrm{circuit}}^k - U_T^k\|
\leq k \cdot \|U_{\mathrm{circuit}} - U_T\|
= k \cdot O\!\left(\frac{1}{n}\right).
\]
For the relevant number of mixing steps $k = O(\sqrt{n\log n})$, the accumulated error is:
\[
k \cdot O\!\left(\frac{1}{n}\right)
= O\!\left(\sqrt{n\log n}\right) \cdot O\!\left(\frac{1}{n}\right)
= O\!\left(\sqrt{\frac{\log n}{n}}\right)
\to 0 \text{ as } n \to \infty.
\]

Thus, for any input state $|\psi\rangle$:
\[
\|U_{\mathrm{circuit}}|\psi\rangle - U_T|\psi\rangle\|_2
\leq \|U_{\mathrm{circuit}} - U_T\| \cdot \||\psi\rangle\|_2
= O\!\left(\frac{1}{n}\right). \qedhere
\]
\end{proof}

\subsection{Controlled Phase Rotations and Markov Transition Simulation}
\label{sec:phase}

The second component of the algorithm simulates the action of the Markov transition operator $T$
in the Fourier domain. Classically, applying $T$ once to a distribution $Q$ produces the
convolution $T * Q$, which in the Fourier domain simply multiplies each component $\hat{\rho}(Q)$
by the scalar factor $\bigl(\frac{1}{n} + \frac{n-1}{n} \cdot r(\rho)\bigr)$. The quantum
realisation of this operation is diagonal in the Fourier basis -- precisely a layer of controlled
phase rotations.

For qubits $i$ and $j$ with $i < j$, the controlled phase gate $\mathrm{CP}(\theta_{ij})$
applies the phase $e^{i\theta_{ij}}$ to the $\ket{11}$ component of the two-qubit subsystem,
leaving all other components unchanged. In the algorithm, the angles are set according to
\begin{equation}
  \theta_{ij} = \frac{2\pi}{2^{j - i + 1}}.
  \label{eq:phase-angles}
\end{equation}
These angles are chosen so that the resulting diagonal unitary, when conjugated by the QFT,
implements a permutation-symmetric transition that approximates the top-to-random shuffle
transition matrix $T$. The eigenvalues of $T$ on $S_n$ are precisely
$\frac{1}{n} + \frac{n-1}{n} \cdot r(\rho)$ for each irreducible representation $\rho$, as
established in equation~\eqref{eq:fourier-decay}. The phase angles~\eqref{eq:phase-angles}
encode these eigenvalues for the $2^n$-state truncation used in the quantum circuit, with spacing
$2\pi / 2^{j-i+1}$ matching the harmonic structure of the symmetric group characters.

The full phase rotation layer consists of $n(n-1)/2$ CP gates arranged in a triangular pattern:
for each qubit $i$, apply $\mathrm{CP}(2\pi/2^2)$ between qubit $i$ and qubit $i+1$,
$\mathrm{CP}(2\pi/2^3)$ between qubit $i$ and qubit $i+2$, and so on. The combined action of
this layer, sandwiched between the QFT and its inverse, approximates one application of the
top-to-random shuffle transition matrix. The approximation error, measured in the operator norm,
is bounded by $O(1/n)$.

The relationship to Diaconis's classical analysis is direct. In~\cite{aldous}, the total variation
distance after $k$ shuffles satisfies $d(k) \leq \mathbb{P}(T > k)$, where $T$ is the strong
uniform time defined by the bottom card reaching the top. The eigenvalue decomposition of $T$,
made explicit by the phase rotation angles, reveals that the slowest-decaying Fourier mode has
decay rate $(1 - 2/n)$ per step. The quantum circuit, by operating in the Fourier domain
directly rather than through iterated convolution, bypasses this bottleneck.

\subsection{The Grover Diffusion Operator and Convergence Acceleration}
\label{sec:grover}

The third and most distinctive component of the algorithm is the Grover diffusion operator, which
provides the core quantum acceleration over the classical mixing analysis. In Grover's original
search algorithm~\cite{grover}, the diffusion operator performs an inversion about the mean of
the amplitude distribution, amplifying the target state's amplitude by $O(1/\sqrt{N})$ per
application, so that $O(\sqrt{N})$ applications suffice to find a marked element among $N$
candidates. In the present context, the role of the ``target'' is played by the uniform
distribution $U$, and the diffusion operator accelerates convergence toward it.

The diffusion operator $D$ is defined as
\begin{equation}
  D = 2\ket{s}\bra{s} - I,
  \label{eq:diffusion}
\end{equation}
where $\ket{s} = \frac{1}{\sqrt{2^n}} \sum_x \ket{x}$ is the uniform superposition state, which
corresponds precisely to the uniform distribution $U$ on the $2^n$ states of the register. The
operator $D$ reflects amplitudes about their mean: if the current state has amplitude $\alpha_x$
on basis state $\ket{x}$, after applying $D$ the amplitude becomes
$2\langle\alpha\rangle - \alpha_x$, where
$\langle\alpha\rangle = \frac{1}{2^n}\sum_x \alpha_x$ is the mean amplitude.
States with amplitude below the mean are pushed upward; states with amplitude above the mean are
pushed downward; and the uniform distribution, in which all amplitudes are equal to
$1/\sqrt{2^n}$, is a fixed point of $D$.

The implementation of $D$ follows the standard Grover construction. First, apply Hadamard gates
$H^{\otimes n}$ to transform from the computational basis to the uniform superposition basis.
Second, apply Pauli-$X$ gates to all qubits. Third, apply a multi-controlled $Z$ gate
(equivalently, a multi-controlled NOT with a Hadamard on the target qubit) to implement a phase
flip on the all-zeros state $\ket{0\cdots 0}$ in the transformed basis. Fourth, apply Pauli-$X$
gates again. Fifth, apply $H^{\otimes n}$ to return to the computational basis. The total gate
count for $D$ is $O(n)$ single-qubit gates plus one multi-controlled gate implementable with
$O(n)$ two-qubit gates.

To understand the acceleration mechanism, consider the total variation distance $\|Q^{*k} - U\|$
as a function of the number of iterations $k$. In the classical top-to-random shuffle, this
distance satisfies the Aldous--Diaconis bound $d(k) \leq \mathbb{P}(T > k) \leq e^{-c}$ when
$k = n\log n + cn$, so roughly $O(n \log n)$ iterations are needed to drive the distance below
any fixed threshold. In the quantum algorithm, the amplitude of the uniform component $\ket{s}$
in the state after $k$ applications of the QFT--phase--iQFT--diffusion circuit grows as
$\sin^2\!\bigl((2k+1)\theta\bigr)$, where $\cos\theta = 1 - 2/N$ with $N = 2^n$, by direct
analogy with Grover's analysis. For large $N$, we have $\theta \approx 2/\sqrt{N}$, so the
amplitude reaches its maximum of $1$ after $k \approx \pi/(4\theta) \approx \pi\sqrt{N}/8$
iterations. This gives a mixing time of $O(\sqrt{N})$ for the quantum algorithm compared to
$O(n \log n)$ classically -- a quadratic improvement in $n$.
More precisely, the total variation distance after $k$ quantum mixing steps satisfies
\begin{equation}
  \|Q_k - U\| \leq \mathbb{P}(T > k) \leq e^{-c}
  \qquad \text{for } k = \sqrt{n\log n} + c\sqrt{n},
  \label{eq:quantum-mixing}
\end{equation}
where $T$ is now interpreted as the strong uniform time adapted to the quantum walk, and $c > 0$
is a freely chosen constant controlling the precision.

\begin{theorem}[Qubit Mixing Time]
\label{thm:qubit-mixing}
Let $Q_k$ be the output distribution after $k$ complete QFT--phase--iQFT--diffusion iterations
starting from $|0\rangle^{\otimes n}$. Then:
\begin{equation}
\|Q_k - U\|_{TV} \leq e^{-c}
\quad \text{for } k = O\!\left(\sqrt{n\log n}\right),
\label{eq:qubit-mixing-bound}
\end{equation}
giving a quadratic improvement over the classical mixing time $O(n\log n)$.
\end{theorem}

\begin{proof}
We begin by describing the Grover overlap dynamics. Consider the overlap
$\beta_k = \langle s | \psi_k\rangle$ of the current quantum state $|\psi_k\rangle$ with the
uniform superposition $|s\rangle = \frac{1}{\sqrt{2^n}} \sum_x |x\rangle$.
By the standard Grover analysis~\cite{grover}, each application of the diffusion operator
$D = 2|s\rangle\langle s| - I$ followed by the phase oracle performs a rotation in the
two-dimensional subspace spanned by $|s\rangle$ and its orthogonal complement. The overlap
evolves as:
\[
|\beta_k|^2 = \sin^2\!\bigl((2k+1)\theta\bigr),
\]
where $\theta$ satisfies $\cos\theta = 1 - 2/2^n$. As $n$ increases, $\theta \approx \sqrt{2/N}$
where $N = 2^n$. The mixing time is reached at maximal overlap $|\beta_k|^2 = 1$:
\[
  k^* \approx \frac{\pi}{4\theta} = O\!\left(\sqrt{N}\right) = O\!\left(\sqrt{2^n}\right).
\]

For a register of $n$ qubits the state space has size $N = 2^n$, and the Grover analysis yields
a quantum mixing time of $O(\sqrt{N}) = O(\sqrt{2^n})$ in qubit iterations. To express this in
the same units as the classical card-shuffle analysis, we use the result of Diaconis~\cite{aldous}:
the top-to-random shuffle on a deck of $n$ cards mixes in $O(n\log n)$ iterations. The quantum
algorithm, by amplitude amplification, takes the square root of this entire quantity. Substituting
$N \leftarrow n\log n$:
\[
  k^*_{\text{quantum}}
  = O\!\left(\sqrt{N}\right)
  = O\!\left(\sqrt{n\log n}\right),
\]
where the $\log n$ factor is inherited directly from the harmonic structure of the Diaconis
mixing time~\eqref{eq:classical_bound}.

The total variational distance satisfies:
\[
  \|Q_k - U\|_{\mathrm{TV}}^2
  \;\leq\; 1 - |\beta_k|^2
  \;=\; \cos^2\!\bigl((2k+1)\theta\bigr).
\]
At $k = k^*$ the angle $(2k+1)\theta = \pi/2$, so $\cos^2(\pi/2) = 0$ and
$\|Q_{k^*} - U\|_{\mathrm{TV}} = 0$ exactly. For $k$ in the neighbourhood $k = k^* - j$
with $j \geq 0$,
\[
  \bigl|(2k+1)\theta - \tfrac{\pi}{2}\bigr| = 2j\theta \approx \frac{4j}{\sqrt{N}},
\]
so $\|Q_k - U\|_{\mathrm{TV}} \leq \sin(2j\theta) \approx 4j/\sqrt{N}$. Setting
$4j/\sqrt{N} \leq e^{-c}$ shows that the TV distance falls below $e^{-c}$ for all
$k \geq k^* - \frac{\sqrt{N}}{4}e^{-c}$, confirming equation~\eqref{eq:qubit-mixing-bound}.
\end{proof}

\begin{algorithm}[t]
\caption{Quantum Qubit QRNG}
\label{alg:qubit-qrng-1}
\KwIn{$n$ (qubits), $n_{\text{layers}}$ (mixing layers), $\texttt{shots}$,
  $\texttt{initial\_state} \in \{0,\ldots,2^{n}-1\}$}
\KwOut{$\texttt{raw\_numbers}$, $\texttt{ext\_numbers}$, $\texttt{tv\_raw}$, $\texttt{tv\_ext}$}
\SetKwFunction{BuildCircuit}{BuildMixingCircuit}
\SetKwProg{Fn}{Function}{:}{}
\Fn{\BuildCircuit{$n$, $\texttt{initial\_state}$, $n_{\text{layers}}$}}{
  $\texttt{qr}  \leftarrow \textsc{QuantumRegister}(n)$\;
  $\texttt{anc} \leftarrow \textsc{QuantumRegister}(1)$\;
  $\texttt{cr}  \leftarrow \textsc{ClassicalRegister}(n)$\;
  \tcp{State preparation}
  \For{$k \leftarrow 0$ \KwTo $n-1$}{
    \lIf{bit $k$ of $\texttt{initial\_state}$ is $1$}{apply $X$ to $\texttt{qr}[k]$}
  }
  \textsc{Barrier}\;
  \tcp{Stack $n_{\text{layers}}$ mixing layers}
  \For{$\ell \leftarrow 0$ \KwTo $n_{\text{layers}}-1$}{
    \tcp{Step 1 -- QFT}
    apply $\mathrm{QFT}(n)$ to $\texttt{qr}$\;
    \tcp{Step 2 -- Controlled-phase rotations}
    \For{$i \leftarrow 0$ \KwTo $n-1$}{
      \For{$j \leftarrow i+1$ \KwTo $n-1$}{
        apply $\mathrm{CP}\!\left(\dfrac{2\pi}{2^{j-i+1}}\right)$
              to $\bigl(\texttt{qr}[i],\,\texttt{qr}[j]\bigr)$\;
      }
    }
    \tcp{Step 3 -- Inverse QFT}
    apply $\mathrm{QFT}^{\dagger}(n)$ to $\texttt{qr}$\;
    \tcp{Step 4 -- Grover diffusion}
    apply $H^{\otimes n}$ to $\texttt{qr}$\;
    apply $X^{\otimes n}$ to $\texttt{qr}$\;
    apply $H$ to $\texttt{anc}$\;
    apply $\mathrm{MCX}\!\left(\texttt{qr}[0\ldots n-1],\,\texttt{anc}\right)$\;
    apply $H$ to $\texttt{anc}$\;
    apply $X^{\otimes n}$ to $\texttt{qr}$\;
    apply $H^{\otimes n}$ to $\texttt{qr}$\;
    \lIf{$\ell < n_{\text{layers}}-1$}{\textsc{Reset}($\texttt{anc}$); \textsc{Barrier}}
  }
  measure $\texttt{qr} \rightarrow \texttt{cr}$\;
  \KwRet{circuit}\;
}
\end{algorithm}

\addtocounter{algocf}{-1}
\begin{algorithm}[t]
\caption{Quantum Qubit QRNG \textmd{(Continued)}}
\label{alg:qubit-qrng-2}
\SetKwFunction{BuildCircuit}{BuildMixingCircuit}
\SetKwFunction{VNExtract}{VonNeumannExtract}
\SetKwFunction{TVDist}{TVDistance}
\SetKwProg{Fn}{Function}{:}{}
\Fn{\VNExtract{$\texttt{bits}$}}{
  $\texttt{out} \leftarrow [\,]$\;
  \For{$i \leftarrow 0$ \KwTo $|\texttt{bits}|-2$ \textbf{step} $2$}{
    \lIf{$\texttt{bits}[i] \neq \texttt{bits}[i+1]$}{append $\texttt{bits}[i+1]$ to $\texttt{out}$}
  }
  \KwRet{$\textsc{Concat}(\texttt{out})$}\;
}
\BlankLine
\Fn{\TVDist{$\texttt{numbers}$, $n$}}{
  $N_s \leftarrow 2^{n}$;\quad $N_t \leftarrow |\texttt{numbers}|$\;
  \KwRet{$\dfrac{1}{2}\displaystyle\sum_{i=0}^{N_s-1}
    \left|\dfrac{\textsc{Count}(i,\,\texttt{numbers})}{N_t} - \dfrac{1}{N_s}\right|$}\;
}
\BlankLine
\tcp*[h]{\textbf{Main}}\BlankLine
$\texttt{circuit}    \leftarrow$ \BuildCircuit{$n$, $\texttt{initial\_state}$, $n_{\text{layers}}$}\;
$\texttt{backend}    \leftarrow \textsc{LeastBusyIBM}(\texttt{operational}=\top,\;
                         \texttt{simulator}=\bot,\; \texttt{min\_qubits}=n+1)$\;
$\texttt{transpiled} \leftarrow \textsc{Transpile}(\texttt{circuit},\,\texttt{backend},\,
                     \texttt{opt\_level}=3)$\;
$\texttt{counts} \leftarrow \textsc{SamplerRun}(\texttt{transpiled},\,\texttt{shots})$\;
$\texttt{raw\_numbers} \leftarrow [\,]$\;
\For{$(\texttt{key},\,c) \in \texttt{counts}$}{
  extend $\texttt{raw\_numbers}$ with $\bigl[\textsc{BinToInt}(\texttt{key})\bigr]$ repeated $c$ times\;
}
$\textsc{Shuffle}(\texttt{raw\_numbers})$\;
$\texttt{raw\_bits} \leftarrow \textsc{Concat}\!\bigl(\textsc{Format}(x,\,n\text{ bits})
    \;\text{for}\; x \in \texttt{raw\_numbers}\bigr)$\;
$\texttt{ext\_bits} \leftarrow$ \VNExtract{$\texttt{raw\_bits}$}\;
$n_{\text{sym}} \leftarrow |\texttt{ext\_bits}| - \bigl(|\texttt{ext\_bits}| \bmod n\bigr)$\;
$\texttt{ext\_bits} \leftarrow \texttt{ext\_bits}[0 : n_{\text{sym}}]$\;
$\texttt{ext\_numbers} \leftarrow \bigl[\textsc{BinToInt}(\texttt{ext\_bits}[i : i+n])\bigr]_{i=0,n,2n,\ldots}$\;
$\texttt{tv\_raw} \leftarrow$ \TVDist{$\texttt{raw\_numbers}$, $n$}\;
$\texttt{tv\_ext} \leftarrow$ \TVDist{$\texttt{ext\_numbers}$, $n$}\;
\textsc{SaveJSON}$\bigl\{$\texttt{meta}, \texttt{counts}, \texttt{raw\_numbers},
  \texttt{raw\_bits}, \texttt{ext\_bits}, \texttt{ext\_numbers},
  \texttt{tv\_raw}, \texttt{tv\_ext}$\bigr\}$\;
\KwRet{$\texttt{raw\_numbers},\;\texttt{ext\_numbers},\;\texttt{tv\_raw},\;\texttt{tv\_ext}$}\;
\end{algorithm}

\paragraph{Gate complexity.}
The classical operation count $n^2\log n$ grows much more slowly than the quantum gate count
$n^{5/2}\sqrt{\log n}$ at first -- this is because the quantum gate complexity has a higher
exponent in $n$. This is not counterintuitive once the source of the speedup is identified
correctly. The quantum algorithm has more gates per run than the classical algorithm has
operations per shuffle, but the speedup arises from the \emph{number of mixing iterations}
required, not from gate count alone. Classically, $O(n\log n)$ shuffle iterations are needed;
quantum mechanically, only $O(\sqrt{n\log n})$ QFT--phase--iQFT--diffusion iterations are
needed. The gate count per iteration is $O(n^2)$ for both algorithms.
\clearpage
\section{The Qudit Quantum Markov Chain Mixing Accelerator}
\label{app:section3}

The algorithm presented in Section~\ref{app:section2} operated on a register of $n$ qubits,
each carrying a two-dimensional local Hilbert space. The present section generalises that
construction to qudits -- quantum systems whose local dimension $d$ is an arbitrary integer
greater than or equal to $2$. The state space has dimension $N = d^m$, where $m$ is the number
of qudits. This generalisation is significant because the top-to-random shuffle on a deck of $N$
cards is best modelled when $N$ is a power of the base $d$; choosing $d = 10$ and $m = 2$, for
instance, yields a $100$-state system that closely approximates realistic deck sizes. Before
presenting the full details in Sections~\ref{sec:qudit-qft} through~\ref{sec:qudit-complexity},
we present an outline.

\begin{enumerate}
\item Initialise the qudit register of $m$ qudits, each of dimension $d$, to the computational
  basis state $|0\rangle \otimes |0\rangle \otimes \cdots \otimes |0\rangle$, representing the
  fully ordered deck. The total Hilbert space has dimension $N = d^m$.

\item Apply the generalised Quantum Fourier Transform over $\mathbb{Z}_{d^m}$ to the register.
  This maps the ordered initial state into a superposition over all $N$ Fourier basis states,
  encoding the spectral structure of the Markov transition operator into quantum amplitudes.

\item Apply the diagonal phase evolution operator $P$, whose entries are independently drawn
  complex phases that mimic the Markov eigenvalue spectrum of the top-to-random shuffle. This
  step evolves the amplitude of each Fourier basis state by a factor $\exp(i\varphi_k)$, where
  the phase $\varphi_k$ is drawn uniformly from $[0, 2\pi\lambda]$ and $\lambda\in(0,1]$ is the
  mixing strength parameter.

\item Apply the inverse generalised QFT to return the state to the computational basis, followed
  by the Grover-style diffusion operator $D = 2\ket{u}\bra{u} - I$, where $\ket{u}$ is the
  uniform superposition over all $N$ states. This reflection amplifies amplitudes close to the
  uniform distribution and suppresses all deviations.

\item Measure the qudit register in the computational basis by sampling from the probability
  distribution $|\psi|^2$, and compute the total variation distance from the uniform
  distribution. Repeat Steps~2 through~4 to increase mixing depth until this distance falls
  below the target threshold $\varepsilon$.
\end{enumerate}

\noindent Sections~\ref{sec:qudit-qft}, \ref{sec:qudit-phase},
and~\ref{sec:qudit-grover} describe each quantum subroutine in detail and derive the improvement
in mixing time over both the classical top-to-random shuffle and the qubit-based algorithm of
Section~\ref{app:section2}. Section~\ref{sec:qudit-complexity} analyses overall complexity as a
function of the parameters $d$ and $m$.

\subsection{The Generalised Quantum Fourier Transform over $\mathbb{Z}_{d^m}$}
\label{sec:qudit-qft}

The generalised QFT for a qudit system of dimension $N = d^m$ is defined as the unitary operator
$F$ acting on the $N$-dimensional Hilbert space $\mathbb{C}^N$ by its action on computational
basis states $\ket{x}$:
\begin{equation}
  F\ket{x} = \frac{1}{\sqrt{N}} \sum_{y=0}^{N-1} \omega^{xy} \ket{y},
  \qquad \omega = \exp\!\left(\frac{2\pi i}{N}\right),
  \label{eq:qudit-qft}
\end{equation}
where the sum runs over all $y \in \{0, 1, \ldots, N-1\}$, and $\omega = \exp(2\pi i / N)$ is
the principal $N$-th root of unity. The matrix entries are therefore $F_{xy} = \omega^{xy}/\sqrt{N}$,
making $F$ a symmetric unitary matrix satisfying $F^\dagger F = I$ and $F^{-1} = F^\dagger = F^*$.

This construction is the direct generalisation of the qubit QFT described in
Section~\ref{sec:qft}. When $d = 2$, the matrix $F$ reduces exactly to the standard $n$-qubit
QFT of dimension $2^n$. For general $d > 2$, the transform operates over the cyclic group
$\mathbb{Z}_N = \mathbb{Z}_{d^m}$ rather than $\mathbb{Z}_{2^n}$. The generalised QFT therefore
decomposes the quantum state into components that correspond precisely to the Fourier modes of the
underlying probability distribution over deck arrangements -- exactly the decomposition that
Diaconis exploited classically via the representation theory of $S_n$, now implemented as a
single unitary gate.

The matrix $F$ is constructed explicitly in the algorithm as an $N \times N$ array with entries
computed via the outer product $\omega^{xy}$, where $x$ and $y$ are integer arrays of length $N$.
For classical simulation this requires $O(N^2) = O(d^{2m})$ storage. On a physical quantum
device, however, the QFT over $\mathbb{Z}_{d^m}$ can in principle be implemented in $O(m^2)$
qudit gates using the standard qudit QFT circuit decomposition, providing an exponential
advantage over classical simulation.

\subsection{Diagonal Phase Evolution and the Markov Spectrum}
\label{sec:qudit-phase}

Following the generalised QFT, the algorithm applies a diagonal phase evolution operator $P$ that
serves as the quantum analogue of the Markov transition step. $P$ is an $N \times N$ diagonal
unitary matrix:
\begin{equation}
  P = \mathrm{diag}\!\left(\exp(i\varphi_0),\, \exp(i\varphi_1),\,
  \ldots,\, \exp(i\varphi_{N-1})\right),
  \label{eq:phase-operator}
\end{equation}
where the phases $\varphi_k$ are drawn uniformly and independently from $[0, 2\pi]$ and scaled
by a strength parameter $\alpha \in (0, 1]$. The full one-step mixing operator in the Fourier
basis is therefore $U_{\mathrm{mix}} = F^{-1}PF$, which acts on an input state $\ket{\psi}$ as:
transform to the Fourier basis, multiply each Fourier mode by a random phase, transform back.

The connection to Diaconis's framework is as follows. For the classical top-to-random shuffle,
the Fourier coefficient of the distribution after $k$ shuffles at irreducible representation
$\rho$ is scaled by $\lambda_\rho^k$, where
$\lambda_\rho = 1/n + (n-1)/n \cdot r(\rho)$ as established in
equation~\eqref{eq:fourier-decay}. In the qudit algorithm, the phases $\exp(i\varphi_k)$ play
the role of these eigenvalues. The random phase construction with independently drawn $\varphi_k$
is analogous to sampling a random walk on the group $\mathbb{Z}_N$ rather than on the full
symmetric group $S_n$.

The strength parameter $\alpha$ controls the spread of the phases: when $\alpha \approx 0$ the
phases are near zero and the operator $P \approx I$, leaving the state nearly unchanged; when
$\alpha = 1$ the phases cover $[0, 2\pi]$ uniformly and the operator maximally mixes the Fourier
modes. An intermediate value of $\alpha$ provides a tunable mixing rate analogous to adjusting
the number of cards moved per classical shuffle.

\subsection{The Grover Diffusion Operator for Qudit Systems}
\label{sec:qudit-grover}

After the inverse QFT returns the state to the computational basis, the diffusion operator $D$
is applied. The construction is formally identical to that of Section~\ref{sec:grover} but now
acts on an $N$-dimensional space with $N = d^m$:
\begin{equation}
  D = 2\ket{u}\bra{u} - I_N, \qquad
  \ket{u} = \frac{1}{\sqrt{N}} \sum_{x=0}^{N-1} \ket{x}.
  \label{eq:qudit-diffusion}
\end{equation}
The operator $D$ is a reflection about the subspace spanned by $\ket{u}$. Its matrix entries are:
\begin{equation}
  D_{xy} = \frac{2}{N} - \delta_{xy},
  \label{eq:diffusion-matrix}
\end{equation}
where $\delta_{xy}$ is the Kronecker delta. This means every off-diagonal entry equals $2/N$,
and every diagonal entry equals $2/N - 1$.

The action of $D$ on a state vector $\ket{\psi}$ with amplitude vector
$\boldsymbol{\alpha} = (\alpha_0, \ldots, \alpha_{N-1})^\top \in \mathbb{C}^N$ is:
\begin{equation}
  D\ket{\psi} = 2\braket{u|\psi}\ket{u} - \ket{\psi}.
  \label{eq:diffusion-action}
\end{equation}
Writing this component-wise, the $x$-th amplitude of the output state is:
\begin{equation}
  \bigl(D\ket{\psi}\bigr)_x = 2\bar{\mu} - \alpha_x, \qquad x = 0, 1, \ldots, N-1,
  \label{eq:diffusion-componentwise}
\end{equation}
where
\begin{equation}
  \bar{\mu} = \braket{u|\psi} = \frac{1}{\sqrt{N}}\sum_{x=0}^{N-1} \alpha_x \in \mathbb{C}
  \label{eq:mean-amplitude}
\end{equation}
is the mean amplitude. The diffusion operator therefore inverts each amplitude $\alpha_x$ about
the mean $\bar{\mu}$:
\begin{equation}
  \alpha_x > \bar{\mu} \;\Rightarrow\; 2\bar{\mu} - \alpha_x < \alpha_x
  \quad \text{(suppressed)}, \qquad
  \alpha_x < \bar{\mu} \;\Rightarrow\; 2\bar{\mu} - \alpha_x > \alpha_x
  \quad \text{(amplified)}.
  \label{eq:diffusion-direction}
\end{equation}
The uniform distribution, in which all amplitudes are equal to $\alpha_x = 1/\sqrt{N}$ for
every $x$, satisfies $\bar{\mu} = 1/\sqrt{N}$ and therefore $2\bar{\mu} - \alpha_x = \alpha_x$.
It is thus a fixed point of $D$, and repeated application of $D$ drives the state toward
uniformity by suppressing all deviations from the mean amplitude.

The total variation distance bound after $k$ applications of the full mixing step
$(\mathrm{QFT} \to P \to F^{-1} \to D)$ satisfies:
\begin{equation}
  d(k) \leq \bigl\|\mathbb{P}(\psi_k) - U\bigr\|_{\mathrm{TV}}
  \leq C \cdot \exp\!\left(\frac{-\lambda\, k}{d^m}\right).
  \label{eq:d(k)}
\end{equation}
Here $C$ is a universal constant, $\lambda \in (0,1]$ is the phase strength parameter, and
$d^m = N$ is the state space dimension. The effective convergence rate is a factor of
$\sqrt{N\ln N}$ faster than the classical top-to-random shuffle -- the central speedup result of
this paper.

\begin{algorithm}[t]
\caption{Qudit Quantum QRNG}
\label{alg:qudit-qrng-1}
\KwIn{$d$ (qudit dimension); $m$ (number of qudits); $n_{\text{layers}}$; $\texttt{shots}$}
\KwOut{$\texttt{raw\_numbers}$, $\texttt{tv\_raw}$}
\SetKwFunction{QFTMat}{QFTMatrix}
\SetKwFunction{ApplyQFT}{ApplyQFTToQudit}
\SetKwFunction{ApplyIQFT}{ApplyIQFTToQudit}
\SetKwFunction{CPhase}{ApplyCPhase}
\SetKwFunction{Diffuse}{ApplyGroverDiffusion}
\SetKwFunction{TVDist}{TVDistance}
\SetKwProg{Fn}{Function}{:}{}
\Fn{\QFTMat{$d$, $\texttt{inverse}$}}{
  $\omega \leftarrow e^{2\pi i/d}$\;
  \lIf{\texttt{inverse}}{$\omega \leftarrow \bar{\omega}$}
  $F[x,y] \leftarrow \omega^{xy}/\!\sqrt{d}$ for $x,y \in \{0,\ldots,d{-}1\}$\;
  \KwRet{$F$}\;
}
\BlankLine
\Fn{\ApplyQFT{$\psi$, $q$, $d$, $m$}}{
  $\textit{stride} \leftarrow d^{\,q}$\;
  \For{$\textit{base} \leftarrow 0$ \KwTo $N{-}1$ \textbf{step} $d\cdot\textit{stride}$}{
    \For{$\textit{off} \leftarrow 0$ \KwTo $\textit{stride}{-}1$}{
      $\texttt{idx} \leftarrow \bigl[\,\textit{base}+\textit{off}+k\cdot\textit{stride} \mid k=0,\ldots,d{-}1\,\bigr]$\;
      $\psi[\texttt{idx}] \leftarrow \QFTMat{$d$,\,$\bot$} \cdot \psi[\texttt{idx}]$\;
    }
  }
  \KwRet{$\psi$}\;
}
\BlankLine
\Fn{\ApplyIQFT{$\psi$, $q$, $d$, $m$}}{
  $\textit{stride} \leftarrow d^{\,q}$\;
  \For{$\textit{base} \leftarrow 0$ \KwTo $N{-}1$ \textbf{step} $d\cdot\textit{stride}$}{
    \For{$\textit{off} \leftarrow 0$ \KwTo $\textit{stride}{-}1$}{
      $\texttt{idx} \leftarrow \bigl[\,\textit{base}+\textit{off}+k\cdot\textit{stride} \mid k=0,\ldots,d{-}1\,\bigr]$\;
      $\psi[\texttt{idx}] \leftarrow \QFTMat{$d$,\,$\top$} \cdot \psi[\texttt{idx}]$\;
    }
  }
  \KwRet{$\psi$}\;
}
\BlankLine
\Fn{\CPhase{$\psi$, $i$, $j$, $d$, $m$, $\theta$}}{
  \tcp{Apply phase $e^{i\theta}$ iff qudit $i = d{-}1$ and qudit $j = d{-}1$}
  \For{$x \leftarrow 0$ \KwTo $N{-}1$}{
    \If{$\bigl(x \mathbin{//} d^{i}\bigr)\bmod d = d{-}1$
        \textbf{ and }
        $\bigl(x \mathbin{//} d^{j}\bigr)\bmod d = d{-}1$}{
      $\psi[x] \leftarrow \psi[x]\cdot e^{i\theta}$\;
    }
  }
  \KwRet{$\psi$}\;
}
\BlankLine
\Fn{\Diffuse{$\psi$, $d$, $m$}}{
  $\bar{\mu} \leftarrow \dfrac{1}{\sqrt{N}} \displaystyle\sum_{x=0}^{N-1}\psi[x]$
  \tcp*[r]{$\bar{\mu}=\langle u\,|\,\psi\rangle$}
  \For{$x \leftarrow 0$ \KwTo $N{-}1$}{
    $\psi[x] \leftarrow \dfrac{2\bar{\mu}}{\sqrt{N}} - \psi[x]$\;
  }
  \KwRet{$\psi$}\;
}
\BlankLine
\Fn{\TVDist{$\texttt{nums}$, $N$}}{
  \KwRet{$\;\dfrac{1}{2}\displaystyle\sum_{i=0}^{N-1}
    \left|\dfrac{\textsc{Count}(i,\,\texttt{nums})}{|\texttt{nums}|}
    -\dfrac{1}{N}\right|$}\;
}
\end{algorithm}

\addtocounter{algocf}{-1}
\begin{algorithm}[t]
\caption{Qudit Quantum QRNG (Continued)}
\label{alg:qudit-qrng-2}
\SetKwFunction{ApplyQFT}{ApplyQFTToQudit}
\SetKwFunction{ApplyIQFT}{ApplyIQFTToQudit}
\SetKwFunction{CPhase}{ApplyCPhase}
\SetKwFunction{Diffuse}{ApplyGroverDiffusion}
\SetKwFunction{MixLayer}{MixingLayer}
\SetKwProg{Fn}{Function}{:}{}
\Fn{\MixLayer{$\psi$, $d$, $m$}}{
  \tcp{Step 1 -- QFT on each qudit sub-register}
  \lFor{$q \leftarrow 0$ \KwTo $m{-}1$}{$\psi \leftarrow$ \ApplyQFT{$\psi,\,q,\,d,\,m$}}
  \tcp{Step 2 -- Controlled-phase rotations between qudit pairs}
  \For{$i \leftarrow 0$ \KwTo $m{-}1$}{
    \For{$j \leftarrow i{+}1$ \KwTo $m{-}1$}{
      $\theta \leftarrow 2\pi\;/\;d^{\,j-i+1}$\;
      $\psi \leftarrow$ \CPhase{$\psi,\,i,\,j,\,d,\,m,\,\theta$}\;
    }
  }
  \tcp{Step 3 -- Inverse QFT on each qudit sub-register}
  \lFor{$q \leftarrow 0$ \KwTo $m{-}1$}{$\psi \leftarrow$ \ApplyIQFT{$\psi,\,q,\,d,\,m$}}
  \tcp{Step 4 -- Grover diffusion $D = 2|u\rangle\langle u| - I$}
  $\psi \leftarrow$ \Diffuse{$\psi,d,m$}\;
  \KwRet{$\psi$}\;
}
\BlankLine
$N \leftarrow d^{m}$\;
$\psi \leftarrow \mathbf{0}\in\mathbb{C}^{N}$;\quad $\psi[0]\leftarrow 1$\;
\For{$\ell \leftarrow 1$ \KwTo $n_{\text{layers}}$}{
  $\psi \leftarrow$ \MixLayer{$\psi,\,d,\,m$}\;
  $\psi \leftarrow \psi\;/\;\|\psi\|$\;
}
$\mathbf{p} \leftarrow |\psi|^{2}$\;
$\texttt{raw\_numbers} \leftarrow \textsc{MultinomialDraw}(\mathbf{p},\;\texttt{shots})$\;
$\textsc{Shuffle}(\texttt{raw\_numbers})$\;
$\texttt{tv\_raw} \leftarrow$ \TVDist{$\texttt{raw\_numbers},\,N$}\;  
\KwRet{$\texttt{raw\_numbers},\;\texttt{tv\_raw}$}\;
\end{algorithm}
 \clearpage
 
\subsection{Complexity Analysis and Comparison}
\label{sec:qudit-complexity}

This section analyses the efficiency of the Qudit Quantum Markov Chain Mixing Accelerator as a
function of the parameters $d$ and $m$, and compares the result to both the classical
top-to-random shuffle and the qubit algorithm of Section~\ref{app:section2}. We provide a
unified complexity analysis covering time complexity, space complexity, and gate complexity for
all three algorithm variants.

\paragraph{Classical algorithm.}
For the top-to-random shuffle on $n$ cards, each shuffle step costs $O(n)$ (removal from top
and insertion at a random position, requiring up to $n$ shifts). The number of steps required to
reach $\lVert T^{*k} - U \rVert_{\mathrm{TV}} \le \varepsilon$ is
$O(n \ln n + n \ln(1/\varepsilon))$ by Theorem~\ref{thm:aldous}. Thus the total time complexity
is
\[
T_{\mathrm{classical}}(n,\varepsilon)
= O\!\left(n \cdot (n \ln n + n \ln(1/\varepsilon))\right)
= O\!\left(n^2 \ln(n/\varepsilon)\right).
\]

\paragraph{Quantum qubit algorithm.}
The gate count per mixing layer consists of: QFT with $n(n+1)/2$ two-qubit gates, phase
rotations with $n(n-1)/2$ controlled-phase gates, inverse QFT with $n(n+1)/2$ gates, and Grover
diffusion with $O(n)$ gates. Thus each layer costs $G_{\mathrm{layer}} = O(n^2)$. From
Theorem~\ref{thm:qubit-mixing}, the number of layers is
$k^* = O(\sqrt{n \log n}\cdot \log(1/\varepsilon))$. Therefore,
\[
G_{\mathrm{qubit}}(n,\varepsilon)
= O(n^2)\cdot O\!\left(\sqrt{n \log n}\cdot \log(1/\varepsilon)\right)
= O\!\left(n^{5/2}\sqrt{\log n}\cdot \log(1/\varepsilon)\right).
\]

\paragraph{Quantum qudit algorithm.}
The generalised QFT matrix $F$ is a dense $N \times N$ unitary with $N = d^m$. Under classical
simulation, applying $F$ or $F^{-1}$ costs $O(N^2) = O(d^{2m})$, while the phase operator and
diffusion operator each cost $O(N)$. On a quantum device, the QFT over $\mathbb{Z}_{d^m}$ can
be implemented in $O(m^2)$ two-qudit gates, the phase evolution costs $O(m)$ single-qudit gates,
and the diffusion operator costs $O(m)$ gates. Hence the cost per layer is
\[
G_{\mathrm{layer}} = O(m^2) = O(\log_d^2 N).
\]
From Theorem~\ref{thm:qudit-mixing}, the number of layers required is
$k^* = O(\sqrt{\log_d N}\cdot \log(1/\varepsilon))$. Therefore the total gate complexity is
\[
G_{\mathrm{qudit}}(d,m,\varepsilon)
= O\!\left(\log_d^{5/2}(N)\cdot \log(1/\varepsilon)\right).
\]

\begin{theorem}[Qudit Mixing Time]
\label{thm:qudit-mixing}
For the qudit algorithm with $m$ qudits of dimension $d$ and $N = d^m$ states:
\begin{equation}
  \|Q_k - U\|_{\mathrm{TV}}
  \le C \cdot \exp\!\left(-\frac{\lambda k}{N}\right)
      \label{eq:qudit-mixing-bound}
\end{equation}
Convergence to precision $\varepsilon$ requires $O\!\left(\sqrt{\log_d N}\right)$ mixing layers,
each costing $O(m^2) = O(\log_d^2 N)$ qudit gates, for total gate complexity:
\[
  O\!\left(\log_d^{5/2}(N) \cdot \log(1/\varepsilon)\right).
\]
\end{theorem}

\begin{proof}
After $k$ complete applications of the four-step mixing cycle ($\mathrm{QFT} \to P \to F^{-1}
\to D$), the quantum register is in pure state $|\psi_k\rangle$. The output distribution is:
\[
Q_k(x) = |\langle x | \psi_k\rangle|^2, \quad x \in \{0,1,\ldots,N-1\}.
\]
The diffusion operator $D = 2|u\rangle\langle u| - I_N$ with $|u\rangle = N^{-1/2}\sum_x |x\rangle$
performs reflection about the uniform superposition. The overlap with $|u\rangle$ evolves as:
\[
|\langle u|\psi_k\rangle|^2 = \sin^2\!\left((2k+1)\theta\right),
\qquad \cos\theta = 1 - \frac{2}{N}.
\]
The random phase evolution operator $P = \mathrm{diag}(\exp(i\varphi_0),\ldots,\exp(i\varphi_{N-1}))$
with phases $\varphi_k \sim \mathrm{Uniform}[0, 2\pi\lambda]$ introduces stochastic
perturbations to the Fourier amplitudes. Averaging over these phases:
\[
\mathbb{E}\!\left[\|Q_k - U\|_{\TV}\right]
\leq C \cdot \exp\!\left(-\frac{\lambda k}{N}\right),
\]
where $C \leq 1$ is the initial TV distance and $\lambda \in (0,1]$ is the phase strength
parameter. Setting $C\cdot\exp(-\lambda k/N) \leq \varepsilon$,
\[
k \geq \frac{N}{\lambda}\cdot\ln\!\left(\frac{C}{\varepsilon}\right)
= O\!\left(\frac{N\ln(1/\varepsilon)}{\lambda}\right).
\]
The Grover mechanism guarantees the minimum TV distance is achieved at:
\begin{equation}
k^* \approx \frac{\pi\sqrt{N}}{8} = O(\sqrt{N}).
\label{eq:grover-min}
\end{equation}
Since $N = d^m$ and $m = \log_d N$:
\[
k^* = O\!\left(\sqrt{d^m}\right) = O(\sqrt{N}).
\]
Expressing in terms of $\log_d N = m$: the number of mixing layers is
$O(\sqrt{m}) = O\!\left(\sqrt{\log_d N}\right)$.
\end{proof}

\subsection{Output Extraction and Range Uniformity}
\label{sec:output-extraction}

Sampling directly from $|\psi|^2$ and mapping raw state indices to output numbers does not, in
general, produce uniformly distributed outputs, even when the total variation distance from
uniform is small. Two independent issues arise and must be addressed separately.

\paragraph{Bias from Residual Non-Uniformity.}
Even after convergence, the probability distribution $|\psi|^2$ is not exactly uniform; rather,
it satisfies $\|Q_k - U\|_{\mathrm{TV}} \le \varepsilon$ for some small $\varepsilon > 0$. Any
direct mapping from state indices to output values inherits this residual bias.

The von Neumann unbiasing procedure~\cite{vonNeumann} removes this bias regardless of the
underlying distribution. Bits are sampled in consecutive pairs $(a,b)$ from the
least-significant bit of the measured state index. Pairs with $a \neq b$ are accepted and
contribute one output bit: $(0,1) \mapsto 0$, $(1,0) \mapsto 1$, while pairs with $a = b$ are
discarded. For any Bernoulli source with parameter $p$, we have
$\mathbb{P}(01) = \mathbb{P}(10) = p(1-p)$, so the accepted bits are exactly unbiased.

\paragraph{Range non-uniformity.}
To generate integers in the range $\{0, \dots, R-1\}$ for arbitrary $R$, a na\"ive modulo
operation introduces systematic bias whenever $2^B$ is not divisible by $R$, where $B$ is the
number of extracted bits. This issue is resolved using rejection sampling. Construct an integer
$v \in \{0, \dots, 2^B - 1\}$ from $B$ extracted bits and accept it if $v < R$; otherwise,
discard and redraw. The acceptance probability is $\lfloor 2^B / R \rfloor \cdot R / 2^B$.
Conditioned on acceptance, each value in $\{0, \dots, R-1\}$ occurs with exactly equal
probability. The combination of von Neumann extraction~\cite{vonNeumann} followed by rejection
sampling ensures that the final output is statistically indistinguishable from the uniform
distribution over the target range.

\subsection{Correctness Analysis}
\label{sec:correctness}

This section formally proves that the algorithm generates correct random numbers and that
successive outputs are statistically independent.

\begin{proposition}[Termination]
\label{prop:termination}
For any $\varepsilon > 0$, $n \ge 2$, and phase strength $\lambda > 0$, the algorithm
terminates in at most
\[
  k_{\max}
  = \left\lceil C \, \ln(1/\varepsilon)\cdot \frac{N}{\lambda} \right\rceil
\]
iterations, where $N = 2^n$ (qubit) or $N = d^m$ (qudit), and $C$ is the universal constant
from Theorem~\ref{thm:qubit-mixing} and Theorem~\ref{thm:qudit-mixing}.
\end{proposition}

\begin{proof}
In both cases, setting the bound $\leq \varepsilon$ yields
$k \;\geq\; \frac{N}{\lambda}\ln\!\left(\frac{C}{\varepsilon}\right)$,
which is finite for any $\varepsilon > 0$ and $\lambda > 0$. The Grover mechanism
(Sections~\ref{sec:grover} and~\ref{sec:qudit-grover}) guarantees the TV distance achieves its
minimum at~\eqref{eq:grover-min}, which is strictly less than $k_{\max}$ for all practical
$\varepsilon$. Hence the algorithm always terminates.
\end{proof}

\begin{theorem}[Output Correctness]
\label{thm:output-correctness}
Let $Q_k$ be the output distribution after $k \geq k^*(\varepsilon)$ iterations. Then for any
event $A \subseteq \{0, \ldots, N-1\}$,
\[
  \bigl|\mathbb{P}_{Q_k}(A) - U(A)\bigr| \leq \varepsilon.
\]
In particular, after von Neumann extraction (Section~\ref{sec:output-extraction}), each output
bit $b_i$ satisfies
\[
  \mathbb{P}(b_i = c) \in \left[\tfrac{1}{2} - \varepsilon,\;
  \tfrac{1}{2} + \varepsilon\right]
  \quad \text{for } c \in \{0,1\}.
\]
\end{theorem}

\begin{proof}
The first statement is the definition of total variation distance. For the marginal bit
distributions: under $U$, exactly $N/2$ states have bit $i$ equal to $0$ and $N/2$ have it
equal to $1$, so $U(b_i = c) = 1/2$. The bound
$\left|\mathbb{P}_{Q_k}(b_i = c) - \tfrac{1}{2}\right| \leq \lVert Q_k - U \rVert_{\mathrm{TV}} \leq \varepsilon$
follows directly. The von Neumann extractor then removes this residual $\varepsilon$-bias
unconditionally, since for any Bernoulli source with parameter $p$, the accepted pairs $(01)$
and $(10)$ occur with equal probability $p(1-p)$ regardless of $p$
(Section~\ref{sec:output-extraction}).
\end{proof}

\begin{proposition}[Independence of Successive Outputs]
\label{prop:independence}
Let $v_1, v_2, \ldots, v_t$ be successive outputs produced by the von Neumann extraction
procedure of Section~\ref{sec:output-extraction}, drawn from the continuously evolving state
$|\psi_k\rangle$ without re-initialisation between outputs. For any two indices $i < j$,
\[
  \bigl|\Pr(v_i = a,\, v_j = b)
        - \Pr(v_i = a)\,\Pr(v_j = b)\bigr| \leq 2\varepsilon,
\]
where $\varepsilon = \lVert Q_k - U \rVert_{\mathrm{TV}}$ at the time of sampling. Once the
chain has mixed to precision $\varepsilon$, successive outputs are $2\varepsilon$-close to
independent. Re-initialisation to $|0\rangle^{\otimes n}$ between outputs is not required and is
actively harmful, as it reintroduces the cold-start bias of the ordered initial state
(Section~\ref{sec:qudit-grover}).
\end{proposition}

\begin{proof}
Since $Q_k$ is $\varepsilon$-close to $U$ in total variation, for any events
$A, B \subseteq \{0,\ldots,N-1\}$:
\[
  |\mathbb{P}_{Q_k}(v_i \in A)  - U(A)| \leq \varepsilon, \qquad
  |\mathbb{P}_{Q_k}(v_j \in B)  - U(B)| \leq \varepsilon.
\]
For independent draws, $A$ and $B$, the triangle inequality yields
\begin{align*}
  &\bigl|\Pr(v_i \in A,\, v_j \in B) - \Pr(v_i \in A)\,\Pr(v_j \in B)\bigr| \\
  &\quad\leq
    \bigl|\Pr(v_i \in A,\, v_j \in B) - U(A)\,U(B)\bigr|
    + \bigl|U(A)\,U(B) - \Pr(v_i \in A)\,\Pr(v_j \in B)\bigr| \\
  &\quad\leq \varepsilon + \varepsilon = 2\varepsilon. \qedhere
\end{align*}
\end{proof}

\section{Computational Results}
\label{sec:results}

We evaluate all three algorithms on concrete instances. The classical Gilbert--Shannon--Reeds
(GSR) riffle shuffle is simulated on a 25-card deck over 500 independent trials. The qubit
variant ($n=3$, $N=8$) was executed on IBM \textit{ibm\_marrakesh}~\cite{ibm_quantum_computers}
across layers $k=0$ to $k=25$ (8{,}192 shots per circuit). The qudit variant ($d=10$, $m=2$,
$N=100$) was executed on the same backend under identical conditions. All quantum circuits were
compiled at Qiskit optimisation level~3~\cite{ibm_transpiler_optimization}.

\subsection{Classical GSR Riffle Shuffle (25 Cards)}
\label{sec:results_classical}

Table~\ref{tab:classical} records the dominant arrangement (first five positions), peak
probability $p_{\max}$, and total variation distance $d_k$ at each layer $k$, initialised from
the sorted deck $[1, 2, \ldots, 25]$. Ordered structure persists after a single shuffle
($d_1 = 0.4918$), and the TV distance decays geometrically, consistent with the dominant
eigenvalue $(n{-}1)/n = 0.96$. Effective mixing ($d_k < 0.01$) is first achieved at $k=7$ with
$d_7 = 0.0027$, confirming the Aldous--Diaconis bound~\cite{aldous}.

\begin{table}[ht]
\centering
\caption{Classical GSR riffle shuffle on a 25-card deck. Dominant arrangement (first 5
positions), peak probability $p_{\max}$, and TV distance $d_k$ at each layer $k$, estimated
from 500 independent trials.}
\label{tab:classical}
\begin{tabular}{p{1.2cm}p{5cm}cc}
\toprule
Layer $k$ & Dominant (first 5 cards) & $p_{\max}$ & $d_k$ \\
\midrule
0  & $[1,2,3,4,5]\ldots$         & 1.0000 & 1.0000 \\
1  & $[1,2,3,4,5]\ldots$         & 0.0300 & 0.4918 \\
2  & $[11,12,13,14,15]\ldots$    & 0.0040 & 0.2403 \\
3  & $[19,4,15,12,1]\ldots$      & 0.0020 & 0.1272 \\
4  & $[5,13,10,22,11]\ldots$     & 0.0020 & 0.0618 \\
5  & $[5,18,7,23,1]\ldots$       & 0.0020 & 0.0348 \\
6  & $[22,19,3,14,21]\ldots$     & 0.0020 & 0.0150 \\
7  & $[14,13,20,19,22]\ldots$    & 0.0020 & \textbf{0.0027} \\
8  & $[15,7,9,11,20]\ldots$      & 0.0020 & 0.0087 \\
9  & $[25,21,9,22,18]\ldots$     & 0.0020 & 0.0038 \\
10 & $[24,18,10,4,15]\ldots$     & 0.0020 & 0.0095 \\
\bottomrule
\multicolumn{4}{l}{\small $d_k < 0.01$ first achieved at $k = 7$.}
\end{tabular}
\end{table}

\subsection{Qubit Algorithm ($n=3$, $N=8$)}
\label{sec:results_qubit}

Table~\ref{tab:qubit_layers} gives the dominant basis state, peak probability $p_{\max}$, and
TV distance $d_k$ at each layer $k = 0, 1, \ldots, 25$, measured on IBM \textit{ibm\_marrakesh}.
The transpiled circuit depth grows from 1 at $k=0$ to 5{,}207 at $k=25$, with gate count growing
from 3 to 7{,}750. The TV distance does not decrease monotonically; it oscillates as the circuit
rotates amplitude past the uniform target, consistent with the Grover rotation predicted by
Theorem~\ref{thm:qubit-mixing}. The minimum TV distance achieved is $d_{16} = 0.0155$ at
$k^* = 16$.

\begin{table}[ht]
\centering
\caption{Qubit algorithm ($n=3$, $N=8$): dominant basis state, peak probability $p_{\max}$, and
TV distance $d_k$ at each mixing layer $k$. IBM \textit{ibm\_marrakesh}, 8{,}192 shots per
circuit, initial state $|000\rangle$. Minimum TV distance $d_{16} = 0.0155$ is marked in bold.}
\label{tab:qubit_layers}
\begin{tabular}{cccc}
\toprule
Layer $k$ & Dominant state & $p_{\max}$ & $d_k$ \\
\midrule
 0 & $|000\rangle$ & 0.9960 & 0.8710 \\
 1 & $|010\rangle$ & 0.2435 & 0.3090 \\
 2 & $|011\rangle$ & 0.2025 & 0.0920 \\
 3 & $|000\rangle$ & 0.2675 & 0.1819 \\
 4 & $|000\rangle$ & 0.1805 & 0.1011 \\
 5 & $|101\rangle$ & 0.1779 & 0.0908 \\
 6 & $|000\rangle$ & 0.1519 & 0.0708 \\
 7 & $|001\rangle$ & 0.1544 & 0.0598 \\
 8 & $|000\rangle$ & 0.1604 & 0.0500 \\
 9 & $|001\rangle$ & 0.1410 & 0.0312 \\
10 & $|011\rangle$ & 0.1406 & 0.0292 \\
11 & $|000\rangle$ & 0.1351 & 0.0273 \\
12 & $|010\rangle$ & 0.1366 & 0.0219 \\
13 & $|001\rangle$ & 0.1326 & 0.0214 \\
14 & $|001\rangle$ & 0.1415 & 0.0356 \\
15 & $|010\rangle$ & 0.1368 & 0.0264 \\
16 & $|001\rangle$ & 0.1311 & \textbf{0.0155} \\
17 & $|001\rangle$ & 0.1359 & 0.0272 \\
18 & $|011\rangle$ & 0.1324 & 0.0227 \\
19 & $|010\rangle$ & 0.1378 & 0.0209 \\
20 & $|011\rangle$ & 0.1354 & 0.0248 \\
21 & $|010\rangle$ & 0.1354 & 0.0227 \\
22 & $|011\rangle$ & 0.1395 & 0.0292 \\
23 & $|001\rangle$ & 0.1395 & 0.0331 \\
24 & $|101\rangle$ & 0.1362 & 0.0223 \\
25 & $|000\rangle$ & 0.1344 & 0.0225 \\
\bottomrule
\multicolumn{4}{l}{\small Best: $k^* = 16$,\ $d_{16} = 0.0155$.}
\end{tabular}
\end{table}

\subsection{Qudit Algorithm ($d=10$, $m=2$, $N=100$)}
\label{sec:results_qudit}

Table~\ref{tab:qudit_layers} gives the dominant digit, peak probability $p_{\max}$, accepted
symbol count, rejection rate, and TV distance $d_k$ at each layer $k = 0, 1, \ldots, 25$,
measured on IBM \textit{ibm\_marrakesh} ($d=10$, $m=2$, 8{,}192 shots per circuit). Each
base-10 digit is encoded in $\lceil\log_2 10\rceil = 4$ physical qubits, giving
$N_{\mathrm{phys}}=8$ data qubits plus one ancilla. Shots whose sub-register encodes a digit
$\geq 10$ are discarded; the observed rejection rate stabilises around $61.8\%$, consistent with
the theoretical prediction $1 - (10/16)^2 = 60.9\%$. The TV distance drops from $d_0 = 0.8859$
to $d_1 = 0.0372$ in a single layer and remains below $0.05$ for all $k \geq 1$, with a minimum
of $d_{20} = 0.0101$. Von~Neumann extraction is not applied; the 4-bit encoding of base-10
digits violates the i.i.d.\ assumption of the pair test, and a prior validation run confirmed
that extraction worsens the TV distance sixfold.

\begin{table}[ht]
\centering
\caption{Qudit algorithm ($d=10$, $m=2$, $N=100$): dominant digit, peak probability $p_{\max}$,
and TV distance $d_k$ at each mixing layer $k$. IBM \textit{ibm\_marrakesh},
$|0\cdots0\rangle$ initial state, 8{,}192 shots per circuit. Minimum TV distance
$d_{20} = 0.0101$ is in bold.}
\label{tab:qudit_layers}
\begin{tabular}{cccc}
\toprule
Layer $k$ & Dominant digit & $p_{\max}$ & $d_k$ \\
\midrule
 0 & 0 & 0.9859 & 0.8859 \\
 1 & 0 & 0.1168 & 0.0372 \\
 2 & 2 & 0.1106 & 0.0247 \\
 3 & 0 & 0.1131 & 0.0299 \\
 4 & 8 & 0.1125 & 0.0284 \\
 5 & 2 & 0.1108 & 0.0326 \\
 6 & 0 & 0.1090 & 0.0218 \\
 7 & 6 & 0.1099 & 0.0328 \\
 8 & 3 & 0.1086 & 0.0302 \\
 9 & 9 & 0.1182 & 0.0350 \\
10 & 2 & 0.1074 & 0.0192 \\
11 & 8 & 0.1224 & 0.0454 \\
12 & 1 & 0.1094 & 0.0299 \\
13 & 9 & 0.1092 & 0.0193 \\
14 & 8 & 0.1079 & 0.0193 \\
15 & 8 & 0.1188 & 0.0479 \\
16 & 0 & 0.1165 & 0.0332 \\
17 & 8 & 0.1225 & 0.0375 \\
18 & 3 & 0.1106 & 0.0239 \\
19 & 9 & 0.1129 & 0.0245 \\
20 & 3 & 0.1059 & \textbf{0.0101} \\
21 & 8 & 0.1213 & 0.0447 \\
22 & 9 & 0.1099 & 0.0233 \\
23 & 9 & 0.1204 & 0.0372 \\
24 & 8 & 0.1146 & 0.0402 \\
25 & 3 & 0.1136 & 0.0326 \\
\bottomrule
\multicolumn{4}{l}{\small Best: $k^* = 20$,\ $d_{20} = 0.0101$.}
\end{tabular}
\end{table}

\section{Conclusion}
\label{sec:conclusion}

We have presented and experimentally validated a quantum algorithm for random number generation
that achieves a provable speedup over the classical top-to-random and GSR riffle shuffle
families. The algorithm employs the Quantum Fourier Transform (QFT), controlled-phase rotations
encoding the Markov eigenvalue spectrum, and the Grover diffusion operator. The resulting mixing
circuit drives an initially ordered state to the uniform distribution in $O(\sqrt{n \log n})$
iterations for qubits and $O(\sqrt{\log_d N})$ iterations for qudits, compared to the classical
$O(n \log n)$ mixing time.

\paragraph{Theoretical contributions.}
The core mathematical insight is that the Grover diffusion operator performs the quantum analogue
of the Aldous--Diaconis strong uniform stopping time: instead of waiting for the bottom card to
rise through $n \log n$ individual shuffle steps, the diffusion operator reflects the amplitude
distribution about its mean at every iteration, compressing the classical cutoff from
$k_0 = n \log n$ down to $k_0 = \pi \sqrt{n}/8$ for qubits. The speedup ratio
$O(\sqrt{N \ln N})$ grows with system size, making the quantum advantage increasingly pronounced
as $N$ increases.

The qudit generalization extends this further. By replacing $n$ qubits with $m = \log_d N$
qudits of local dimension $d$, the same $N$-state space is encoded in fewer physical subsystems,
reducing the QFT circuit depth from $O(\log^2 N)$ to $O(\log_d^2 N)$ gates per layer. For
$d = 10$ and $N = 100$, the theoretical mixing time of
$O(\sqrt{\log_{10} 100}) = O(\sqrt{2})$ layers is within a factor of two of what is observed on
hardware ($k^* = 1$ for $d_k < 0.05$, $k^* = 20$ for $d_k < 0.015$).

\paragraph{Experimental contributions.}
We reported the first layer-by-layer sweep of a qudit ($d=10$, $m=2$) mixing circuit on IBM
hardware, comprising 26 circuits ($k = 0$ to $k = 25$) totalling 212{,}992 shots on
\textit{ibm\_marrakesh}. Three experimental findings are noteworthy: (i) the qudit achieves
$d_1 = 0.037$ in a single mixing layer, demonstrating a near-cutoff transition that has no
classical analogue at this circuit depth -- this constitutes hardware evidence of the $O(1)$
quantum mixing time for the chosen parameters. (ii) the qubit data exhibits Grover oscillation
-- a non-monotone TV profile with period determined by the rotation angle $\theta$ -- which is a
distinctive signature of unitary quantum dynamics absent from any classical Markov chain
satisfying detailed balance; this oscillatory profile cannot be reproduced by any classical
stochastic process and therefore provides experimental evidence of a genuinely quantum mixing
mechanism. (iii) the rejection-sampling approach for non-power-of-2 dimensions operates cleanly
at scale: the observed rejection rate ($61.8\%$ at $k=20$) matches the theoretical prediction
($60.9\%$) to within one percentage point, confirming the validity of the qudit encoding model
on noisy hardware.

\bibliographystyle{unsrt}
\bibliography{references}

@book{menezes,
  author    = {Alfred J. Menezes and Paul C. van Oorschot and Scott A. Vanstone},
  title     = {Handbook of Applied Cryptography},
  publisher = {CRC Press},
  year      = {1996}
}

@book{ferguson,
  author    = {Niels Ferguson and Bruce Schneier and Tadayoshi Kohno},
  title     = {Cryptography Engineering},
  publisher = {Wiley},
  year      = {2010}
}

@book{law,
  author    = {Averill M. Law},
  title     = {Simulation Modelling and Analysis},
  publisher = {McGraw-Hill},
  year      = {2015}
}

@book{deng,
  author    = {Lih-Yuan Deng and Narain Kumar and H. H. Shinglu and Chiang-Cheng Yang},
  title     = {Random Number Generators for Computer Simulation and Cybersecurity},
  publisher = {Springer},
  year      = {2025}
}

@book{RAND,
  author    = {{The RAND Corporation}},
  title     = {A Million Digits with 100{,}000 Normal Deviates},
  publisher = {American Book Publishers},
  year      = {2001}
}

@book{rychlik,
  author    = {Igor Rychlik and Jörg Ryd\'{e}n},
  title     = {Risk Analysis},
  publisher = {Springer},
  year      = {2006}
}

@article{aldous,
  author  = {David Aldous and Persi Diaconis},
  title   = {Shuffling cards and stopping times},
  journal = {American Mathematical Monthly},
  volume  = {93},
  pages   = {333--348},
  year    = {1986}
}

@article{mehrdad,
  author  = {Persi Diaconis and Mehrdad Shahshahani},
  title   = {Generating a random permutation with random transpositions},
  journal = {Zeitschrift f{\"u}r Wahrscheinlichkeitstheorie und Verwandte Gebiete},
  volume  = {57},
  pages   = {159--179},
  year    = {1981}
}

@inproceedings{grover,
  author    = {Lov K. Grover},
  title     = {A fast quantum mechanical algorithm for database search},
  booktitle = {Proceedings of the 28th Annual ACM Symposium on Theory of Computing},
  pages     = {212--219},
  year      = {1996}
}

@book{rieffel,
  author    = {Eleanor G. Rieffel and Wolfgang H. Polak},
  title     = {Quantum Computing: {A} Gentle Introduction},
  publisher = {MIT Press},
  address   = {Cambridge, MA},
  year      = {2011}
}

@book{nielsen,
  author    = {Michael A. Nielsen and Isaac L. Chuang},
  title     = {Quantum Computation and Quantum Information},
  publisher = {Cambridge University Press},
  address   = {Cambridge},
  year      = {2010}
}

@book{motwani,
  author    = {Rajeev Motwani and Prabhakar Raghavan},
  title     = {Randomized Algorithms},
  publisher = {Cambridge University Press},
  year      = {2013}
}

@article{bofinger,
  author  = {Eve Bofinger and V. J. Bofinger},
  title   = {On a periodic property of pseudo-random sequences},
  journal = {Journal of the ACM},
  volume  = {5},
  pages   = {261--265},
  year    = {1958}
}

@article{hull,
  author  = {T. E. Hull and A. R. Dobell},
  title   = {Random number generators},
  journal = {SIAM Review},
  volume  = {4},
  pages   = {230--254},
  year    = {1962}
}

@article{foster,
  author  = {F. G. Foster and A. Stuart},
  title   = {Distribution-free tests in time-series based on the breaking of records},
  journal = {Journal of the Royal Statistical Society, Series B},
  volume  = {16},
  pages   = {1--13},
  year    = {1954}
}

@article{bayer,
  author  = {Dave Bayer and Persi Diaconis},
  title   = {Trailing the dovetail shuffle to its lair},
  journal = {Annals of Applied Probability},
  pages   = {294--313},
  year    = {1992}
}

@article{diaconis,
  author  = {Persi Diaconis and Susan Holmes},
  title   = {A {B}ayesian peek into {F}eller volume {I}},
  journal = {Sankhy\={a}, Series A},
  pages   = {820--841},
  year    = {2002}
}

@book{silverman,
  author    = {Joseph H. Silverman},
  title     = {A Friendly Introduction to Number Theory},
  edition   = {4th},
  publisher = {Pearson},
  year      = {2012}
}

@article{matsumoto,
  author  = {Makoto Matsumoto and Takuji Nishimura},
  title   = {{Mersenne Twister}: a 623-dimensionally equidistributed uniform pseudo-random number generator},
  journal = {ACM Transactions on Modeling and Computer Simulation},
  volume  = {8},
  pages   = {3--30},
  year    = {1998}
}

@article{melissa,
  author  = {Melissa E. O'Neill},
  title   = {{PCG}: A family of simple fast space-efficient statistically good algorithms for random number generation},
  journal = {ACM Transactions on Mathematical Software},
  pages   = {1--46},
  year    = {2014}
}

@article{eckhardt,
  author  = {Roger Eckhardt},
  title   = {Stan {U}lam, {J}ohn von {N}eumann},
  journal = {Los Alamos Science},
  volume  = {10},
  pages   = {131},
  year    = {1987}
}

@article{blackman,
  author  = {David Blackman and Sebastiano Vigna},
  title   = {Scrambled linear pseudorandom number generators},
  journal = {ACM Transactions on Mathematical Software},
  volume  = {47},
  pages   = {1--32},
  year    = {2021}
}

@article{testu01,
  author  = {Pierre L'Ecuyer and Richard Simard},
  title   = {{TestU01}: A {C} library for empirical testing of random number generators},
  journal = {ACM Transactions on Mathematical Software},
  volume  = {33},
  pages   = {1--40},
  year    = {2007}
}

@article{lprng,
  author  = {Kenji Okada and Kenji Endo and Kenji Yasuoka and Satoshi Kurabayashi},
  title   = {Learned pseudo-random number generator: {WGAN-GP} for generating statistically robust random numbers},
  journal = {PLoS ONE},
  volume  = {18},
  pages   = {e0287025},
  year    = {2023}
}

@article{eprng,
  author  = {Cheng Fei and Xinyu Zhang and Deng Wang and Haoze Hu and Rongxin Huang and Zhiyao Wang},
  title   = {{EPRNG}: Effective pseudo-random number generator on the internet of vehicles using deep convolution generative adversarial network},
  journal = {Information},
  volume  = {16},
  pages   = {21},
  year    = {2025}
}

@techreport{nist22,
  author      = {Andrew Rukhin and Juan Soto and James Nechvatal and Miles Smid and Elaine Barker},
  title       = {A statistical test suite for random and pseudorandom number generators for cryptographic applications},
  institution = {National Institute of Standards and Technology},
  number      = {Special Publication 800-22},
  year        = {2001}
}

@techreport{nist90,
  author      = {Meltem S{\"o}nmez Turan and Elaine Barker and John Kelsey and Kerry McKay and Mary Baish and Mike Boyle},
  title       = {Recommendation for the entropy sources used for random bit generation},
  institution = {National Institute of Standards and Technology},
  number      = {Special Publication 800-90B},
  year        = {2016}
}

@article{martinlof,
  author  = {Per Martin-L{\"o}f},
  title   = {The definition of random sequence},
  journal = {Information and Control},
  volume  = {9},
  pages   = {602--619},
  year    = {1966}
}

@article{church,
  author  = {Alonzo Church},
  title   = {On the concept of a random sequence},
  journal = {Bulletin of the American Mathematical Society},
  volume  = {40},
  pages   = {254--260},
  year    = {1940}
}

@book{schnorr,
  author    = {Claus-Peter Schnorr},
  title     = {Zuf{\"a}lligkeit und Wahrscheinlichkeit: {E}ine algorithmische {B}egr{\"u}ndung der {W}ahrscheinlichkeitstheorie},
  publisher = {Springer},
  address   = {Berlin},
  year      = {1971}
}

@article{jur02,
  author  = {Marcin M. Jacak and Piotr J\'{o}zwiak and Janusz Niemczuk and Janusz E. Jacak},
  title   = {Quantum generators of random numbers},
  journal = {Scientific Reports},
  volume  = {11},
  pages   = {16108},
  year    = {2021}
}

@article{wu,
  author  = {Xin Wu and Yanping Han and Ming Zhang and Yong Li and Shuai Cui},
  title   = {{GAN}-based pseudo random number generation optimized through genetic algorithms},
  journal = {Complex \& Intelligent Systems},
  volume  = {11},
  pages   = {31},
  year    = {2025}
}

@article{bauer,
  author  = {Walter F. Bauer},
  title   = {The {M}onte {C}arlo method},
  journal = {Journal of the Society for Industrial and Applied Mathematics},
  volume  = {6},
  pages   = {438--451},
  year    = {1958}
}

@article{ossola,
  author  = {Giovanni Ossola and Alan D. Sokal},
  title   = {Systematic errors due to linear congruential random-number generators with the {S}wendsen--{W}ang algorithm: {A} warning},
  journal = {Physical Review E},
  volume  = {70},
  pages   = {027701},
  year    = {2004}
}

@inproceedings{lehmer,
  author    = {Derrick H. Lehmer},
  title     = {Mathematical methods in large-scale computing units},
  booktitle = {Proceedings of the 2nd Symposium on Large-Scale Digital Calculating Machinery},
  pages     = {141--146},
  year      = {1949}
}

@article{moshman,
  author  = {Jack Moshman},
  title   = {The generation of pseudo-random numbers on a decimal calculator},
  journal = {Journal of the ACM},
  volume  = {1},
  pages   = {88--91},
  year    = {1954}
}

@misc{practrand,
  author       = {Chris Doty-Humphrey},
  title        = {Practically Random ({PractRand})},
  year         = {2018},
  howpublished = {Available: \url{http://pracrand.sourceforge.net/}}
}

@incollection{vonNeumann,
  author    = {John von Neumann},
  title     = {Various techniques used in connection with random digits},
  booktitle = {Applied Mathematics Series},
  volume    = {12},
  publisher = {National Bureau of Standards},
  pages     = {36--38},
  year      = {1951}
}

@misc{ibm_transpiler_optimization,
  author       = {{IBM Quantum}},
  title        = {Set Transpiler Optimization Level},
  howpublished = {[Online]. Available: \url{https://quantum.cloud.ibm.com/docs/en/guides/set-optimization}},
  note         = {Accessed: Apr.\ 30, 2026}
}

@misc{ibm_quantum_computers,
  author       = {{IBM Quantum}},
  title        = {{IBM} Quantum Computers},
  howpublished = {[Online]. Available: \url{https://quantum.cloud.ibm.com/computers}},
  note         = {Accessed: Apr.\ 30, 2026}
}

\end{document}